\documentclass[11pt]{article}
\usepackage{times}
 
 \usepackage{subcaption}
 \usepackage{graphicx}

 \def\colorful{0}
 \usepackage{times}
 \usepackage{amsfonts,amsmath,amssymb, mathtools, color,float,graphicx,verbatim}
 \usepackage[ruled,vlined,linesnumbered]{algorithm2e}
 \usepackage{hyperref} 
\usepackage{multirow}  
 \usepackage{amsthm,amsfonts,amsmath,color,float,graphicx,verbatim}
 \usepackage{multirow}
 \usepackage{amsmath, color, enumitem}
 \usepackage{framed}
 \usepackage{nicefrac}
 \usepackage{thm-restate}
 
\hypersetup{final}
 
 \usepackage{enumitem}
 \usepackage[capitalize]{cleveref}

 \def\nnewcolor{1}
 \ifnum\nnewcolor=1
 
 \fi
 \ifnum\nnewcolor=0
 
 \fi
 
 \ifnum\colorful=1

 \else

 \fi

\usepackage[dvipsnames]{xcolor}
\hypersetup{
    colorlinks,
    linkcolor={red},
    citecolor={ForestGreen},
    urlcolor={Violet}
}

 \newtheorem{theorem}{Theorem}[section]
 
 \newtheorem{lemma}[theorem]{Lemma}

 \newtheorem{remark}[theorem]{Remark}

 \newtheorem{definition}[theorem]{Definition}

 \newcommand\numberthis{\addtocounter{equation}{1}\tag{\theequation}}
 
 \newcommand{\UU}{\mathcal{U}}

\def\E{\ensuremath{\mathrm{\mathbf{E}}}}
\def\Pr{\ensuremath{\mathrm{\mathbf{Pr}}}}
\def\Var{\ensuremath{\mathrm{\mathbf{Var}}}}
\def\Cov[2][]{\ensuremath{\mathrm{\mathbf{Cov}}}}

 \title{Testing Properties of Multiple Distributions with Few Samples}

\author  {Maryam Aliakbarpour
 	\thanks{Massachusetts Institute of Technology, Cambridge, MA 02139. Email:~\url{maryama@mit.edu}. Research supported by MIT-IBM Watson AI Lab (Agreement No. W1771646), NSF grants IIS-1741137, and CCF-1733808.}
 	\and Sandeep Silwal
 	\thanks{Massachusetts Institute of Technology, Cambridge, MA 02139. Email:~\url{silwal@mit.edu}. Research supported by the NSF Graduate Research Fellowship under Grant No. 4000054330.}}

 \def\accept{{\fontfamily{cmss}\selectfont accept}\xspace}
 \def\reject{{\fontfamily{cmss}\selectfont reject}\xspace}

\begin{document}
 	
 	\maketitle
 	
 	\setcounter{page}{0}
 	
 	\thispagestyle{empty}

\begin{abstract}
We propose a new setting for testing properties of distributions while receiving samples from several distributions, but few samples per distribution. Given samples from $s$ distributions, $p_1, p_2, \ldots, p_s$, we design testers for the following problems: (1) Uniformity Testing: Testing whether all the $p_i$'s are uniform or $\epsilon$-far from being uniform in $\ell_1$-distance (2) Identity Testing: Testing whether all the $p_i$'s are equal to an explicitly given distribution $q$ or $\epsilon$-far from $q$ in $\ell_1$-distance, and (3) Closeness Testing: Testing whether all the $p_i$'s are equal to a distribution $q$ which we have sample access to, or $\epsilon$-far from $q$ in $\ell_1$-distance. By assuming an additional natural condition about the source distributions, we provide sample optimal testers for all of these problems. 
\end{abstract}

\section{Introduction} \label{sec:intro}

Statistical tests are a crucial tool in scientific endeavors to analyze data: We routinely model data to be a set of samples from an unknown distribution, and use statistical tests to infer or verify the properties of the underlying distribution. While these tests typically operate under the assumption that data points are drawn from a {\em single} underlying distribution, in  applications, usually the data is gathered from multiple sources. Furthermore in many situations, it is the case that the dataset contains only a few data points from each source. For example, an online shop may have the purchase history of thousands of customers while each customer may shop at the store a small number of times. Alternatively, a medical dataset might record the lifestyle behaviors of patients of a particular disease while only having few data points from any specific demographic (such as age). 

On the other hand, data that comes from multiple sources may result in a dataset consisting of a collection of unconnected and unrelated data points. For example, it might not be possible to derive any meaningful conclusions from a dataset that contains the blood pressure of patients with heart diseases, Alzheimer patients, and healthy individuals. However, if there is some consensus among the sources, we may be able to make reasonable inferences based on the data. Therefore, an important question to ask is: how can we 
mathematically model agreement among the sources such that it is possible to design testers with theoretical guarantees?

In this work, we propose a framework for hypothesis testing, one of the most fundamental problems in statistics, while allowing for the underlying data to be drawn from multiple distributions (sources) 
and only receiving ``few" samples from each distribution. More specifically, we study the following problem: suppose we have $s$ source distributions, $p_1, \ldots, p_s$. We have a distribution $q$ (hypothesis), and we aim to distinguish between the case where all the source distributions are equal to $q$ and the case where all the source distributions are far from $q$. We propose a \emph{structural condition} in order to model the agreement among the sources to enable us to draw meaningful conclusions.

Our \emph{structural condition} requires all the sources to have the same preference for every element, meaning that for each domain element $x$, either all the sources assign higher probability than the ``speculated" probability, $q(x)$, or all of them assign lower probabilities. However, the sources can go arbitrarily higher or lower than $q(x)$ as long as they stay on the same side of the $q(x)$. For example, suppose one has tried several prize wheels (lottery machines) in a casino. The player spins the wheel and expects to receive one of the prizes uniformly. Given the results of each spin, our goal is to test whether all the machines were fair (i.e., selecting the prize uniformly), or they are far from being fair. In this case, we can naturally assume that if the machines are unfair, the house will assign a lower probability to the expensive prizes, and higher probability to cheap ones. Another example is political affiliation at a local vs. national level. Suppose a political party polls its constituents in a district about their opinion on the most crucial policy and compares it with national polls. It is natural to assume that the policies of national interest will receive the same responses in different districts. On the other hand, if a policy affects the district positively (or negatively), members of the district are more (less) likely to pick them. It is worth noting that if no structural condition is assumed, the problem becomes vacuous even in the simplest cases. The main issue is that two completely different sets of distributions may result in identical set of samples. For example, suppose each $p_i$ is a singleton distribution on a random element $x \in [n]$. If we draw one sample from each distribution, the samples we obtain will be indistinguishable from the samples that are i.i.d.\@ from a uniform distribution over $[n]$. See Section~\ref{sec:Motivation} for more elaboration.

Given our agreement condition, we consider three different cases for our hypothesis $q$: $(i)$ \emph{Uniformity testing:} $q$ is uniform. $(ii)$ \emph{Identity testing (goodness of fit):} $q$ is explicitly known. $(iii)$ \emph{Closeness testing (equivalence test):} $q$ is accessible through samples. We require each source distribution to provide \emph{exactly one} sample for uniformity and one sample in expectation for identity and closeness testing.
We develop sample optimal testers for all these three problems. In fact, the sample complexity of our testers is exactly equal to the standard versions of these problems when samples are drawn from a single source. These results lead to the belief that our agreement condition provides the same power as the standard setting for designing the testers while operating under a weaker assumption.

Our sample complexity upper bounds are achieved by using variants of testers previously used for distribution testing in the case where samples are drawn from a fixed distribution. The challenge however, lies in analyzing these testers in our more general setting with multiple sources. The sample complexity lower bounds follow directly from the single distribution setting.  For a full description of our contributions and approaches, see Section~\ref{sec:contributions}.

\subsection{Necessity of modeling multiple sources} We might hypothesize that data points drawn from different distributions can be thought of as coming from some `average' or `aggregated' distribution. Indeed, we know by de Finetti's theorem that an infinite sequence of exchangeable random variables is actually drawn from a mixture of product distributions. In other words, there is some latent variable such that conditioning on this variable, all the samples are independently drawn from one probability distribution. However in the case that we have finitely many samples (or equivalently finitely many sources), de Finetti type theorems only hold up to some approximation error and in the case where the number of samples is sublinear in the domain size, we give a family of distributions where the sequence of random variables with each sample drawn from a different distribution cannot be seen as a mixture of product distributions. This result implies that modeling data as samples from a single distribution is not sufficient when multiple sources are involved. 
See Section~\ref{sec:definetti_fail} for more information.

\subsection{Comparison with other models} Studying properties of a collection of distributions has been studied prior to our work in~\cite{LRR13, ABR16, DiakonikolasK:2016}. These papers consider two primary models for sampling a collection of distributions. In the first model, which is called the \emph{query model}, the user can query each distribution and receive a sample from it. In the second model, which is called the \emph{sampling model}, the user does not get to choose the source distribution, but the user receives a pair $(i,j)$ which can be interpreted as a sample from the collection: the first element $i$ indicates that distribution $i$ was selected with a probability proportional to some (known or unknown) weight, and then $j$ is a sample drawn from the $i$-th distribution.

There are few differences between our model and two models listed above. In these models, there is no limit on the number of samples that can come from a distribution. This is in contrast to our setting where every distribution contributes only one sample in expectation. On the other hand, in these two models, there is no agreement condition imposed between the different distributions, and their goal is to distinguish if all the distributions are equal or their \emph{average distance} from a single distribution is at least $\epsilon$. Considering the average distance essentially turns this problem into testing closeness of a distribution over the domain $[n]\times[s]$ which requires more samples.

While our problems are inherently different, 
none of the results in the papers cited above  solve the problems we consider using a sublinear number of queries. In fact in some regime of the parameters, their algorithms draw $\omega(1)$ samples (even in expectation). In some special case, where the number of samples per distribution is $\Theta(1)$ in expectation, the sample complexity of their algorithm is greatly larger than ours. In particular, suppose we have $s$ distributions over a domain of size $n$ and we draw $m$ samples from them in total. In the query model, the provided algorithms pick a few distributions and draw $O(n^{2/3})$ samples from them which is in contrast to our requirement of one sample per distribution. Moreover for the sampling model, the optimal algorithm needs $m = O(\sqrt{n s}/\epsilon^2 + n^{2/3}s^{1/3}/\epsilon^{4/3})$ samples in total.
Roughly speaking, if the number of distributions is asymptotically smaller than $n$, i.e., $s = o(n)$, then certainly the number of samples, $m$, has to be $\omega(s)$ meaning that we need more than $\Theta(1)$ samples per distribution. 
On the other hand, if the number of  distributions, $s$, is $\Omega(n)$, then the number of samples, $m$, has to be $\Omega(n/\epsilon^2)$ which is drastically larger than our sample complexity, $O(\sqrt{n}/\epsilon^2 + n^{2/3}/\epsilon^{4/3})$.

In~\cite{TianKV17, vinayak19a}, the authors consider a similar setting as our paper. In their setting, they have $N$ distributions over the domain of size two. Each distribution is determined by a parameter which indicates the probability of the first domain element, and the algorithm receives $t$ samples from each distribution. However, these papers consider a very different problem compared to ours as their goal is to optimally learn the histogram of the parameters with approximation error as a function of $t$.

\subsection{Other related work}
Distribution property testing is a framework for investigating properties of a distribution(s) upon receiving samples. This framework was first introduced in~\cite{GGR98, Batu:2000}, and it is part of the broader topic of hypothesis testing in statistics~\cite{NeymanP, lehmann2005testing}. In this framework, we wish to determine if one or more unknown distributions satisfy a certain property or are `far' from satisfying the property. The goal is to obtain an algorithm, or tester, for this task that has the optimal sample complexity. Since its introduction, several properties have been considered. See~\cite{Rub12, canonne2015survey,GoldreichBook17} for a survey of results.

The problems of testing uniformity, identity, and closeness of distributions have first been considered in \cite{GR00,Batu:2000, BatuFFKRW}where it is assumed that samples are always drawn from a fixed distribution. Many subsequent work improved on their results, and eventually testers with optimal sample complexities of $\Theta(\sqrt{n}/\epsilon^2)$ for identity and uniformity testing, and $\Theta(n^{2/3}/\epsilon^{4/3} + \sqrt{n}/\epsilon^2)$ for closeness testing were obtained. See~\cite{Valiant:2008, Paninski:08, VV11, ChanDVV14, DiakonikolasKN14, ADK15, DiakonikolasK:2016,Goldreich2016TheUD, DiakonikolasGPP16, DiakonikolasGPP18, BlaisC19,Batu2017GeneralizedUT}. For a survey of techniques used for these problems, see~\cite{canonne2015survey}.

\subsection{Organization}
We start with definitions and preliminaries in Section~\ref{sec:preliminaries}. In Section \ref{sec:unif_testing}, we study uniformity testing with samples from multiple sources. 
In Section \ref{sec:id_testing}, we study identity testing with non-identically drawn samples. In Section \ref{sec:closeness_testing}, we study closeness testing with non-identically drawn samples. Finally, we prove Theorem \ref{thm:definettiex} in Appendix \ref{sec:definetti}.

\section{Preliminaries} \label{sec:preliminaries}
\subsection{Notation and Definitions}
We use $[n]$ to denote the set $\{1, \cdots, n\}$. We consider discrete distributions over $[n]$, which are non-negative functions $p:[n] \rightarrow [0,1]$ such that $\sum_{i \in [n]} p(i) = 1$. We let $p(i)$ denote the probability assigned to element $i \in [n]$ by a distribution $p$ and for a set $A \subseteq [n]$, we define $p(A) = \sum_{i \in A} p(i)$. For $q \ge 1$, the $\ell_q$-norm of distribution $q$ is defined as $\| p \|_q = (\sum_{i \in [n]} p(i)^q )^{1/q}$. Given two distributions $p$ and $p'$, the $\ell_q$-distance between them is defined as the $\ell_q$-norm of the vector of their differences: $\| p - p'\|_q = (\sum_{i \in [n]} |p(i)-q(i)|^q)^{1/q}$. The total variation distance of two distributions $p$ and $p'$ is defined as $\|p - p'\|_{TV} = \sup_{A} |p(A)-p'(A)|$ which is known to be equal to $\|p-p'\|_1/2$. We say that two distributions $p$ and $p'$ are $\epsilon$-far in $\ell_q$-distance if $\|p-p'\|_q \ge \epsilon$. Otherwise, we say that $p$ and $p'$ are $\epsilon$-close in $\ell_q$-distance. In this paper, we primarily focus on $\ell_1$-distance. We denote the uniform distribution over $[n]$ by $\UU_n$. Also, we refer to a Poisson random variable with parameter $\lambda$ as $\text{Poi}(\lambda)$.

\subsection{The Structural Condition} \label{sec:Motivation}
We introduce the structural condition used in our multiple source distribution testing setting. This condition models the assumption that the different sources have an {\em agreement} of the preferences which we explain earlier.

\begin{definition}[Structural Condition]\label{def:structural}  Given a sequence of distributions $p_1, p_2, \cdots$ over $[n]$ and another distribution $q$ over $[n]$, we say that $\{p_i\}_{i \ge 1}$ satisfy the structural condition if there exist sets $A \subset [n], B = [n] \setminus A$, such that for all the $p_i$'s,
\begin{align*} 
		p_i(j) \geq q(j) & \quad \quad \forall j \in A \,,\\ 
		p_i(j) \leq q(j) & \quad \quad \forall j \in B.
	\end{align*}
	Note that we \textbf{do not} assume the knowledge of what the sets $A$ and $B$ are, we only assume the existence of the two sets. 
\end{definition}
\subsubsection{Alternative agreement conditions}
To motivate Definition \ref{def:structural}, our \emph{structural condition}, we focus on the problem of uniformity testing. In the usual setting of uniformity testing, we are given sample access to a \emph{single} unknown probability distribution $p$ over $[n]$, and we wish to determine if $p$ is equal to $\UU_n$  or if $\|p-\UU_n\|_1 \ge \epsilon$. 

The most general relaxation of the single source assumption is to allow each sample to be drawn from a possibly different distribution. In particular, we wish to distinguish the completeness case, where each sample is i.i.d.\@ from $\UU_n$, from the soundness case, where sample $i$ is drawn independently from $p_i$, and $p_i$ and $p_j$ are not necessarily the same for $i \ne j$, and $\|p_i - \UU_n\|_1 \ge \epsilon$ for all $i$. 
By using the relation between the $\ell_1$-norm and the total variation distance, this general setting can be written in the following way in the soundness case which we require the total variation distance between every $p_i$ and the uniform distribution to be at least $\epsilon/2$:
\begin{equation}\label{eq:def1}
\min_{i} \max_{A \subseteq [n]} | p_i(A)-\UU_n(A)| \ge \epsilon/2. 
\end{equation}
However, we cannot hope to drive meaningful conclusions in this setting. Consider the case where each $p_i$ is a singleton distribution on a random element $x \in [n]$. If we draw one sample from each distribution, the samples we obtain will be indistinguishable from the samples that are i.i.d.\@ from a uniform distribution over $[n]$.
A natural strengthening of \eqref{eq:def1} is to assume that in the soundness case, not only each distribution is different from $\UU_n$ on some set $A$, as we had above, but all the $p_i$'s are far from $\UU_n$ on the \emph{same} set $A$. This can be written as:
\begin{equation}\label{eq:def2}
\max_{A\subseteq [n] }\min_{i} | p_i(A)-\UU_n(A) | \ge \epsilon/2. 
\end{equation}
(Note that the min and max are switched from Equation~\eqref{eq:def1}). In other words, there is some fixed set $A$ such that $p_i$ and $\UU_n$ are assigning very different probability mass to the set $A$. However, this assumption is still too weak to support uniformity testing in sublinear time. The main reason is that we can come up with $s$ distribution satisfying Equation~\eqref{eq:def2}, but the samples drawn from them look the same as uniform distribution. In general, for testing a symmetric property (i.e., a property that does not depend on the labeling of the elements), e.g., uniformity, we only consider the number of repetition in the sample set. The main sources of information is how many elements repeated $t$ many times in the sample set. In the single distribution setting, these information is related to the moments of the underlying distribution, and it is known that distributions with the similar moments requires a lot of samples to tell them apart~\cite{RaskhodnikovaRSS:2007, Valiant08, wu2015chebyshev}.  

Consider the following example where we have $s < n$ distributions, and each distribution $p_i$ is supported on $[1,i] \subset [n]$. For $i \in [s]$ The distribution $p_i$ assigns the following probability to the domain element $x \in [n]$.
$$
p_i(x) =
\begin{cases}
\frac{1}n & \text{if } x < i\\
1- \frac{i-1}n & \text{if } x = i\\
0& \text{if } x > i
\end{cases}.
$$
Let $A$ be the set of elements that all the $p_i$'s assign zero probability to them: $\{s+1, s+2, \ldots, n\}$. Clearly in our example Equation~\eqref{eq:def2} holds for a parameter $\epsilon < 1$. As long as $s \le (1-\epsilon)n$ since $\|p_i-\UU_n\|_1/2 \geq |p_i(A) - \UU_n(A)| \geq (n - s)/n \geq \epsilon$ for all $i$. Now, the probability that samples $i$ and $j$, where $i < j$, are equal is
$$ \frac{i-1}{n^2} + \left(1 - \frac{i-1}n \right)\frac{1}n = \frac{1}n  $$
which is exactly the probability of a collision between two different samples in the completeness case. Furthermore, for any $k \le s \le (1-\epsilon)n$, we can compute the probability that any $k$ samples $i_1 <  \cdots < i_k$ match. Due to the support of $p_{i_1}$, we know that this quantity is precisely
$$ \frac{i_1-1}{n^{k}} + \left( 1 - \frac{i_1-1}n \right) \frac{1}{n^{k-1}} = \frac{1}{n^{k-1}} $$
which is exactly the probability that any $k$ samples all match if all samples are drawn from the uniform distribution. Therefore with some generalized notion of the moments, the set of distributions in the above example match the first $O(n)$ moments of the uniform distribution. Due to the matching of these moments, we cannot hope to test uniformity (or any other symmetric property). Hence, a stronger structural condition than \eqref{eq:def2} is needed to allow testing in our setting. In this work, we proposed a natural strengthening of the assumption \eqref{eq:def2}, given in Definition \ref{def:structural}, which is enough to perform uniformity testing, along with other hypothesis testing problems. This is elaborated in Section \ref{sec:contributions}.

	\subsection{Our Contributions} \label{sec:contributions}

\subsubsection{Uniformity testing with multiple sources} In our multiple source distributions setting for uniformity testing, we have $s$ distributions, $p_1, \ldots, p_s$, and each distribution provides {\em exactly one} sample. Our goal is to distinguish the following cases with probability at least $2/3$\footnote{Note that the constant $2/3$ is arbitrary here. One can boost the accuracy to $1-\delta$ for an arbitrary small $\delta$ by increasing the number of samples (distributions) by a $O(\log \delta^{-1})$ factor.}:
\begin{itemize}
	\item {\textbf{Completeness case:}} $p_1, p_2, \cdots$ are all uniform on $[n]$. 
	\item{\textbf{Soundness case:}} $p_1, p_2, \cdots$ are all $\epsilon$-far from uniform on $[n]$ in $\ell_1$-distance.
\end{itemize}
Furthermore, we impose that in the soundness case, the distributions $\{p_i\}_{i=1}^s$ satisfy the \textit{structural condition} given in Definition \ref{def:structural} with $ q$ being $\UU_n$, the uniform distribution. That is in the soundness case, all the distributions have mass at least $1/n$ on the elements in $A$ and at most $1/n$ on the elements in $B$ for some sets $A$ and $B$ that are \textbf{unknown} to us. Note that the structural condition trivially holds in the completeness case when all the $p_i$'s are the same distribution. Therefore, we can think of our setting as a generalization of uniformity testing. 

We show that the standard collision-based algorithm used in the single distribution case of uniformity testing (\cite{Goldreich2011, Batu:2000, DiakonikolasGPP16}) is able to distinguish the completeness and the soundness case in our multiple sources setting. The statistic that we calculate is the number of pairwise collisions among the samples. We show that in the completeness case, there are `few' collisions among the samples whereas in the soundness case, we see `many' collisions. 
The main challenge is the analysis of this statistic in the soundness case, since the distributions $p_1, p_2, \cdots$ are not necessarily the same.

In the completeness case that all the $p_i$'s are equal to some distribution $p$, the collision statistic is equal to a multiple of the $\ell_2$-norm of $p$.   We proceed similarly by introducing a more general notion of $\ell_2$-norm in our setting. In addition, we argue that our statistic is sufficiently concentrated by calculating its variance. We generalize the tight variance analysis of \cite{DiakonikolasGPP16}, which shows that the collision based tester is optimal in the single source setting. Again if all the $p_i$'s are equal to some distribution $p$, as is the case in the single source uniformity testing setting, the variance is related to the $\ell_3$-norm of $p$. In our case where the $p_i$'s are not necessarily the same, we introduce a generalized notion of $\ell_3$-norm and relate it to our notion of $\ell_2$-norm. This argument relies on Maclaurin's inequality. Altogether, our analysis shows that we can perform uniformity testing in our setting using $O(\sqrt{n}/\epsilon^2)$ samples, which is optimal since the standard single source uniformity testing is a special case of our setting, has a known sample complexity lower bound of $\Omega(\sqrt{n}/\epsilon^2)$~\cite{Paninski:08}. This result is presented in Section~\ref{sec:unif_testing}.

\subsubsection{Identity testing with multiple sources} We now describe identity testing in the multiple source distributions setting. We first assume that we explicitly know some fixed distribution $q$ over $[n]$. We suppose we have $s$ distributions, $p_1, \ldots, p_s$. Our goal then is to distinguish the following cases with probability at least $2/3$:
\begin{itemize}
	\item {\textbf{Completeness case:}} $p_1, p_2, \cdots, p_{s}$ are identical to $q$
	\item{\textbf{Soundness case:}} $p_1, p_2, \cdots, p_{s}$ are all $\epsilon$-far from $q$ in $\ell_1$-distance.
\end{itemize}
Furthermore, we impose that the distributions $\{p_i\}_{i=1}^s$  and $q$ satisfy the \textit{structural condition} given in Definition \ref{def:structural}. For identity testing with multiple sources, we use a generalization of the poissonization method used in many distribution testing problem (see \cite{canonne2015survey}): we assume that we receive $\text{Poi}(1)$ samples, as opposed to one sample from each distribution $p_i$ that we had in the uniformity case. Clearly, each distribution provides one sample in expectation, and with high constant probability, no distribution provides more than $O(\log s)$ samples. 

In standard single distribution identity testing, a modified version of Pearson's $\chi^2$-test statistic is picked to calculate the expected value of $\|q-p\|_2^2$, where $q$ is the known distribution and all samples are from $p$~\cite{ValiantV14, ChanDVV14, ADK15, DiakonikolasK:2016}. In our case, we generalize this approach and give a new statistic, again a modified version of Pearson's $\chi^2$-test, which calculates a variant of the $\ell_2$-distance between our known distribution $q$ and the distributions that our samples come from. 

In particular, if we take $s$ samples, we show that the expected value of our statistic is $\left\| \sum_{j = 1}^s \vec{\mathbf{e}}_j \right \|_2^2$, where $\vec{\mathbf{e}}_j$ is a vector in $\mathbb{R}^n$ where the $x$-th entry is $|p_j(x)-q(x)|$ for $x \in [n]$. Note that one can think of this quantity as a generalization of $\|q-p\|_2^2$. We then show that the sample complexity of distinguishing the soundness and completeness cases for our generalized identity testing depends on $\|q\|_2$, the $\ell_2$-norm of the known distribution. The main technical issue is to analyze the variance of our statistic which is challenging since each sample can come from a possibly distinct distribution. Finally, we show how to reduce $\|q\|_2$ using a `flattening' scheme adapted from \cite{DiakonikolasK:2016} that only enlarges the domain size by a constant factor which results in the sample complexity of $O(\sqrt{n}/\epsilon^2)$ which is optimal since the standard single distribution uniformity testing is a special case of our generalized identity testing, and it requires $\Omega(\sqrt{n}/\epsilon^2)$ samples~\cite{Paninski:08}.

\noindent\textbf{Remark on Goldreich's reduction from identity to uniformity testing:\quad}
There is a reduction from identity testing to uniformity testing given in \cite{goldreich_reduction} in which Goldreich gives mappings $F_1$ and $F_2$ such that if $p$ is $\epsilon$-far from $q$, then $F_2(F_1(p))$ is a distribution that is $O(\epsilon)$-far from the uniform distribution over a larger domain of size $m = n/\gamma$ where $\gamma$ is a parameter of the reduction. This reduction works partially in our case. Denote $F = F_2 \circ F_1$. Then we can check that $F(p_i)$ is $O(\epsilon)$-far from $q$ for every $p_i$ in the soundness case. Furthermore, if $x \in [n]$ is also in $A$, then the domain elements corresponding to $x$ in $[m]$ are all at least $1/m$ and similarly, if $x \in [n]$ is in the set $B$, the domain elements corresponding to $x$ in $[m]$ are all at most $1/m$. In particular, $F$ maps the set $A \subseteq [n]$ to a set $A' \subseteq [m]$ that has the same properties as $A$ and similarly, $F$ maps the set $B \subseteq [n]$ to another subset $B' \subseteq [m]$.

The issue in applying this reduction to our setting is with $F_1$. In particular, $F_1(p_i)$ increases the domain size from $[n]$ to possibly $[n+1]$, and there is no guarantee if the domain elements in $m$ corresponding to $n+1$ will be in $A'$ or $B'$. In particular, it could be that for some $p_i$'s, these domain elements will be in $A'$ and for other $p_i$'s, these domain elements can possibly be in $B'$. This could create potential `cancellations' that hide collisions when observing samples from $F(p_i)$. To fix this, we would have to make sure these domain elements don't have much probability mass, which leads to letting $\gamma = O(\epsilon)$. This ultimately leads to a sub optimal query complexity in terms of $\epsilon$ for identity testing. Therefore, we do not pursue this approach.

\subsubsection{Closeness testing with multiple sources} We now describe our generalized version of closeness testing. We assume that we have access to two streams of samples. In the first stream, all samples are i.i.d.\@ from some fixed distribution $q$ over $[n]$ that is unknown to us. In the second stream, samples are drawn independently from distributions $p_1, \ldots, p_s$ where $p_i$ and $p_j$ are not necessarily the same distribution for $i \ne j$. Our goal then is to distinguish the following cases with probability at least $2/3$:
\begin{itemize}
	\item {\textbf{Completeness case:}} $p_1, p_2, \cdots, p_s$ are identical to $q$
	\item{\textbf{Soundness case:}} $p_1, p_2, \cdots, p_s$ are all $\epsilon$-far from $q$ in $\ell_1$-distance.
\end{itemize}
We also impose that the distributions $\{p_i\}_{i=1}^s$  and $q$ satisfy the \textit{structural condition} given in Definition \ref{def:structural}. Note that the structural condition trivially holds in the completeness case.

Our approach to closeness testing with multiple sources is very similar to our approach for identity testing above. We make use of the poissonization method. In particular, we draw $\text{Poi}(s)$ samples from distribution $q$. Also, we take $\text{Poi}(1)$ samples from each of the distributions $p_i$, so in total we have $\text{Poi}(s)$ samples from the distributions $\{p_i\}_{i=1}^s$. Furthermore, we use a (different) modified version of Pearson's $\chi^2$-test proposed in~\cite{ChanDVV14, DiakonikolasK:2016} and show that the expected value of our statistic is $\left\| \sum_{j = 1}^s \vec{\mathbf{e}}_j \right \|_2^2$ where the vector $\vec{\mathbf{e}}_j$ is the same as in the identity testing case above. With a careful analysis of the statistic, in contrast with~\cite{DiakonikolasK:2016}, we show that the sample complexity only depends on the $\ell_2^2$-norm of the $q$. \footnote{Similar analysis has appeared in~\cite{ADKR19} before this work.} Finally, we use a (randomized) `flattening' scheme from \cite{DiakonikolasK:2016} which results in the sample complexity of $O(n^{2/3}/\epsilon^{4/3} + \sqrt{n}/\epsilon^2)$  which is optimal since there is a known lower bound of $\Omega(\max(n^{2/3}/\epsilon^{4/3}, \sqrt{n}/\epsilon^2))$ for the single distribution setting of closeness testing~\cite{DiakonikolasK:2016}. Using the same techniques, we also obtain a tester which uses asymptotically different number of samples from $q$ compared to the number of sources (known as testing with unequal-sized samples). See Remark~\ref{remark:unequal-sized-closeness} for more details.

\subsubsection{Failure of de Finetti's Theorem with sublinear number of samples}\label{sec:definetti_fail}
An infinite sequence $X_1, X_2, \ldots$ of random variables is called exchangeable if for all $m \ge 1$, the distribution of the sequence $X_1, \ldots, X_m$ is identical to the distribution of $X_{\sigma(1)}, \ldots, X_{\sigma(m)}$ for  any permutation $\sigma$ on $m$ elements. de Finetti's theorem states that any {\em  infinite} exchangeable sequence is a mixture of product distributions. In other words, there exists a probability measure $\mu$ such that conditioned on $\mu$, $X_1, X_2, \ldots$ can be seen as i.i.d. samples from a distribution. 

Similarly, a finite sequence $X_1, \ldots, X_m$ is called exchangeable if all the permutations of the sequence have the same distribution. If an exact version of de Finetti's theorem were to hold for finite sequences, our new setting where each sample can come from a different distribution reduces to the known setting where all the samples are i.i.d.\@ (since an algorithm can turn the samples it sees into an exchangeable sequence by randomly permuting the samples). However, all the known finite versions of de Finetti's type theorems have an error term which roughly states that finite exchangeable sequences are only {\em approximately} close to mixtures of product distributions (see \cite{main_exchangeability, definetti_alt, definettibook}). 

In Section \ref{sec:definetti}, we give an example of a finite sequence of random variables that falls in the soundness case of our setting of uniformity testing with multiple sources that is $\Omega(1)$-far from any mixture of product distributions. More precisely, our theorem, Theorem \ref{thm:definettiex},  tells us that it is not always possible to approximate a finite exchangeable sequence $X_1, \ldots, X_s$ arbitrarily well by a mixture of product distributions. This suggests that it is not possible to use de Finetti's theorem in our setting and therefore, more refined tools are needed rather than a hammer like de Finetti's theorem. More formally, our theorem is the following.

\begin{restatable}{theorem}{definettiex}
	\label{thm:definettiex}
	Let $s = O(\sqrt{n})$ be the number of samples required by Algorithm \ref{alg:uniformity-tester} for $\epsilon = 1/3$. There exists an exchangeable sequence $X_1, \ldots, X_s$ such that $X_i$ is drawn from distribution $q_i$ which are all supported in $[n]$ and satisfy $\|q_i-\UU_n\|_1 \ge 1/3$ for all $i$. Furthermore, $\{q_i\}_{i=1}^s$ all satisfy the structural condition given in Definition \ref{def:structural} with $q = \UU_n$. Let $P$ denote the distribution of the sequence $X_1, \ldots, X_s$. Then $P$ is $\Omega(1)$-far in $\ell_1$-distance from any mixture of product distributions.
\end{restatable}

The proof of Theorem \ref{thm:definettiex} uses ideas from Diaconis and Freedman in \cite{main_exchangeability}. For other variants and finite extension of de Finetti's theorem, see \cite{definettibook}.

\section{Uniformity Testing with Multiple Sources} \label{sec:unif_testing}

We now present our algorithm, \textsc{Uniformity-Tester}, for uniformity testing with multiple sources. We show the standard collision based statistic, introduced in ~\cite{GR00, BFRSW}, is a sufficient statistic to distinguish whether all sources are uniform or all sources are $\epsilon$-far from uniform in our multiple sources setting. 
The collision statistic is selected based on a simple observation: if we draw two samples from a distribution, the probability that these two samples are equal (also known as a \emph{collision}) is lowest when the distribution is uniform. Thus, the number of  pairwise collisions tends to be ``small" if the samples are drawn from a uniform distribution. We show that this observation still holds in our setting. Our algorithm takes $s$ samples (for a parameter $s$ which we determine later) and calculates the number of pairwise collisions in the samples. Then, it compares the number collisions to a threshold, $\tau$, which we specify later. If the number of collisions is less than $\tau$, we infer the sources are uniform and output \accept; otherwise, we infer the sources are far from uniform, and output \reject. We present our approach in~Algorithm~\ref{alg:uniformity-tester} along with the main theorem, Theorem~\ref{thm:unif}, which proves the correctness of our algorithm. 



\begin{algorithm}[t]
	\SetKwInOut{Input}{Input}
	\SetKwInOut{Output}{Output}
	\Input{$n$, $\epsilon$, one sample from each of $p_1, p_2, \ldots, p_s$}
	\Output{\accept or \reject}
	\DontPrintSemicolon
	$s \gets \frac{c_1 \, \sqrt n}{\epsilon^2}$
	\;
	Take $s$ samples $X_1, \cdots, X_s$. 
	\;
	For each $1 \le i < j \le s$, let $\sigma_{ij}$ be the indicator variable for the event $X_i = X_j$.
	\;
	$\tau \gets \frac{1 + \epsilon^2/16}{n}$ 
	\;
	$Z \gets \sum_{i < j} \sigma_{ij} / \binom{s}2$ \;
	\If{$Z \ge \tau$ }{Output \reject and abort.}
	Output \accept.
	\caption{$\textsc{Uniformity-Tester}$}
	\label{alg:uniformity-tester}
\end{algorithm}

\begin{theorem}[Correctness of \textsc{Uniformity-Tester}]\label{thm:unif}
	There exists a constants $c_1$ independent of $n$ such that  the following statements hold with probability $2/3$:
	\begin{itemize}
		\item \textbf{Completeness case:} \textsc{Uniformity-Tester} outputs  \accept if each of the $s$ distributions $p_1, \ldots, p_s$ are uniform.
		\item \textbf{Soundness Case:} \textsc{Uniformity-Tester} outputs 
		\reject if the $p_i$'s are $\epsilon$-far from the uniform distribution (i.e., $\|p_i - \UU_n \|_1 \ge \epsilon$) and $\{p_i\}_{i=1}^s$ satisfy the \textit{structural condition} of Definition \ref{def:structural} with $q = \UU_n$.
	\end{itemize}
\end{theorem}

\begin{remark}
	The sample complexity of Algorithm~\ref{alg:uniformity-tester} is optimal due to the lower bound of $\Omega(\sqrt{n}/\epsilon^2)$ for testing uniformity in the standard single source setting presented in~\cite{Paninski:08}.
\end{remark}

\noindent\textbf{Overview of the proof:} \quad To prove the correctness of \textsc{Uniformity-Tester}, we analyze the statistic $Z$ which is the number of collisions in the sample set:
$$ Z = \frac{1}{\binom{s}2} \sum_{1 \le i < j \le s} \sigma_{ij}$$
where $\sigma_{ij}$, for $i < j$, is the indicator that sample $i$ is equal to sample $j$. The algorithm outputs \accept or \reject by comparing the statistic $Z$ to a threshold $\tau$. Our goal is to show $Z$ is below the threshold in the completeness case and above the threshold in the soundness case. To do so, we first compute the expectation of $Z$ and then show a sufficiently strong concentration around its expectation by bounding the variance of $Z$. By a careful selection of the number of samples and the threshold $\tau$, we can prove with high probability that $Z$ is on the desired side of the threshold, and consequently the correctness of the algorithm.

\begin{proof}[Proof of Theorem \ref{thm:unif}] We start by setting the parameters: Let the threshold $\tau$ be $(1 + \epsilon^2/16)/n$. Define $\alpha$ to be the solution to $\E[Z] = (1+\alpha)/n$. Let the number of samples, $s$, to be $c_1 \sqrt n /\epsilon^2$ for a sufficiently large constant $c_1$.
	
	Note that in the completeness case, all samples are coming from the uniform distribution. In this case, $Z$ is analyzed in \cite{GR00, Batu:2000, DiakonikolasGPP16}, so we know the expected value of the statistic is as follows:
	$$\E[Z] =  \|\UU_n\|_2^2 = \frac{1}{n} \,.$$ 
	Furthermore, the variance of $Z$ is bounded from above as below (see Lemma~2.3 in \cite{DiakonikolasGPP16}):
	$$\Var[Z] \leq \Theta\left(\frac{s^2 \cdot \|\UU_n\|_2^2 + m^3 \left(\|\UU_n\|_3^3 - \|\UU_n\|_2^4 \right)}{\binom{s}2^2}\right) \leq \Theta\left( \frac{1}{ns^2}\right)\,.$$
	Now, by Chebyshev's inequality, we can bound the probability that $Z$ become larger than the threshold as follows
	\begin{align*}
		\Pr\left[Z \geq \tau \right]	 
		& \leq \Pr\left[\left|Z - \E[Z]\right| \geq \frac{\epsilon^2}{16\, n} \right] 
		\leq \Theta\left(\frac{n}{\epsilon^4 \, s^2} \right) \leq \frac 1 3
		\,
	\end{align*}
	where the last inequality holds for a sufficiently large constant $c_1$ and having $s = c_1 \sqrt n /\epsilon^2$ which proves the correctness of the completeness case.
	
	The main challenge of this proof is to analyze the soundness case when the $p_i$'s are potentially different. We first give a lower bound for the expected value of $Z$. We begin by providing an intuitive overview of our approach. In the soundness case, we can compute that the expected value of the indicator random variable for a collision between the $i$-th and the $j$-th sample is given by
	\begin{equation} \label{eq:l2norm}
	\E[\sigma_{ij}] = \sum_{x \in [n]} p_i(x)p_j(x) 
	\end{equation}
	where $p_i$ and $p_j$ are the distributions that sample $i$ and $j$ are respectively drawn from. One can think of Equation\@ \eqref{eq:l2norm} as a generalization of $\|p \|_2^2$ when two distributions are involved. To bound Equation\@ \eqref{eq:l2norm} from below, we make use of the \textit{structural condition}. Namely, we can define the error terms 
	\begin{equation}\label{eq:error_terms}
	\begin{split}
	e_i(x) &= p_i(x)-\frac{1}n \quad \quad  \forall x \in A \\
	e_i(x)& = \frac{1}n - p_i(x) \quad \quad  \forall x \in B \,.
	\end{split}
	\end{equation}
	We know that 
	$$ \sum_{x \in A} e_i(x) = \sum_{x \in B} e_i(x) \,. $$ 
	In fact, the above quantities are half the $\ell_1$-distance between $p_i$ and the uniform distribution. We define $e_j(x)$ similarly for $p_j$, and the above identity similarly holds for the $e_j$'s as well. Using these equations, we show in Lemma \ref{lem:ev_lb} that 
	$$  \sum_{x \in [n]} p_i(x)p_j(x) = \frac{1}n + \sum_{x \in [n]}e_i(x)e_j(x). $$
	Recall that our goal is to show that the expected number of collisions in the soundness case is substantially larger than $\binom{s}2/n$. Thus, we desired to bound find a lower bound for the second term in the right hand side above. However, since $p_i$ and $p_j$ are not necessarily the same distribution, it could be the case that for a fixed pair $i,j$ we have $\sum_{x \in [n]}e_i(x)e_j(x) = 0$ which is what we would expect if $p_i$ and $p_j$ were both uniform. Thus, instead of bounding $\sum_{x \in [n]}e_i(x)e_j(x)$ for each pair $i$ and $j$, we show that the sum of these terms over \emph{all the pairs} $i < j$ is $\Theta(\epsilon^2)$. More formally, we have the following lemma.
	
	\begin{restatable}{lemma}{lemEVLB} \label{lem:ev_lb}
		Let $\{p_i\}_{i=1}^s$ be distributions over $[n]$ that are all $\epsilon$-far from $\UU_n$ in $\ell_1$-distance, and satisfy the \textit{structural condition} given in Definition \ref{def:structural} with $q = \UU_n$. Let $X_i$ be drawn independently from $p_i$ for all $1 \le i \le s$. Let $\sigma_{ij}$ be the indicator variable for the event $X_i = X_j$ and define $Z = \sum_{i < j}\sigma_{ij}/\binom{s}2$. Then the following estimate holds
		$$ \E[Z] \ge \frac{1 + \,\epsilon^2/8}n \,.$$
	\end{restatable}
	
	\begin{proof}
		Recall the error terms which we defined in Equation~\eqref{eq:error_terms}: 
		\begin{equation*}
			\begin{split}
				e_i(x) &= p_i(x)-\frac{1}n \quad \quad  \forall x \in A \,,\\
				e_i(x)& = \frac{1}n - p_i(x) \quad \quad  \forall x \in B \,.
			\end{split}
		\end{equation*}
		We start by giving a convenient representation of $\E[\sigma_{ij}]$ in terms of the error terms:
		\begin{equation}\label{eq1}
		\E[\sigma_{ij}]  = \frac{1}n + \sum_{x \in [n]} e_i(x)e_j(x)\,.
		\end{equation}
		To prove the above equation, observe that since the sum of the probabilities in any discrete distribution is one, we have:
		\begin{equation}	\label{eq:error_terms_identity}
		\sum_{x \in A} e_i(x) = \sum_{x \in B} e_i(x) \,
		\end{equation}
		and similar for the $e_j$'s. All of the $e_i(x)$'s and the $e_j(x)$'s  are non-negative for any domain element $x$ by definition and the \textit{structural condition}. Thus, we can obtain:
		\begin{align*}
			\E[\sigma_{ij}] & = \sum_{x \in [n]} p_i(x)p_j(x) \\
			&= \sum_{x \in A}\left(e_i(x) + \frac{1}n \right)\left(e_j(x) + \frac{1}n \right) + \sum_{x \in B} \left(\frac{1}n - e_i(x) \right)\left(\frac{1}n  - e_j(x)\right)	
			\\ & 
			= \frac{|A|}{n^2} + \frac{1}n \sum_{x \in A} e_i(x) + \frac{1}n \sum_{x \in A} e_j(x) + \sum_{x \in A} e_i(x)e_j(x)
			\\ & 
			+ \frac{|B|}{n^2} - \frac{1}n \sum_{x \in B} e_i(x) - \frac{1}n \sum_{x \in B} e_j(x) + \sum_{x \in B} e_i(x)e_j(x)
			\,.
		\end{align*}
		Using Equation~\eqref{eq:error_terms_identity}, it is clear that the sum of two middle terms above are zero:
		
		$$ \frac{1}n \sum_{x \in A} e_i(x) -  \frac{1}n \sum_{x \in B} e_i(x) = 0 
		\,,\quad\mbox{and}\quad\quad\quad\quad
		\frac{1}n \sum_{x \in A} e_j(x) -  \frac{1}n \sum_{x \in B} e_j(x) = 0 \,,$$
		which implies the desired identity we claimed in Equation~\eqref{eq1}:
		\begin{equation*} 
			\E[\sigma_{ij}]= \frac{|A| + |B|}{n^2} + \sum_{x \in [n]} e_i(x)e_j(x) = \frac{1}n +  \sum_{x \in [n]} e_i(x)e_j(x)\,.
		\end{equation*}
		Using this identity for all $\binom{s}2$ pair of samples, yields to the following:
		\begin{equation} \label{eq2}
		\E\left[\binom{s}2 Z\right] = \E\left[\sum_{j < i} \sigma_{ij} \right] = \sum_{j < i} \sum_{x \in [n]} p_i(x)p_j(x) = \dbinom{s}2 \frac{1}n +  \sum_{j < i} \sum_{x \in [n]} e_i(x)e_j(x).
		\end{equation}
		Now, we focus on the second term on the right hand side of the equation above. We can compute that
		\begin{equation}\label{eq:sum_error_terms}
		\sum_{ j < i } \sum_{x \in B} e_i(x)e_j(x)  = \frac{1}2 \left( \underbrace{ \sum_{x \in B} \left( \sum_{i=1}^s e_i(x) \right)^2}_{\mbox{first term}} - \underbrace{\sum_{i=1}^s \sum_{x \in B} e_i(x)^2}_{\mbox{second term}} \right)\,.
		\end{equation} 
		To fine a lower bound $\E[Z]$, we find a lower bound for the first term and an upper bound for the second term in the right hand side above. 
		
		\vspace{2mm}\noindent\textbf{Lower bound for the first term:\quad} Note that if $x \in B$, by definition, $e_i(x)$ is at most $1/n$. On the other hand, $\sum_{x \in B} e_i(x)$ is half of the $\ell_1$-distance between $p_i$ and the uniform distribution. Define $\epsilon'_i$ to be $\|p_i - \UU_n\|_1 = 2\,\sum_{x \in B} e_i(x)$. Clearly, $\epsilon'_i$ is at least $\epsilon$. Then, we have the following lower bound for the size of $B$: 
		$$
		|B| \cdot \frac{1}{n} \geq \sum_{x \in B} e_i(x) = \frac{\epsilon'_i}{2}, \quad \Rightarrow \quad |B| \geq \frac{\epsilon'_i n}{2} \geq \frac{\epsilon n}{2} \,.
		$$
		Therefore, it follows that for any $i$, we have
		$$ \sum_{x \in B} e_i(x)^2 \le \frac{1}{n^2} \cdot \frac{\epsilon'_i n}2 = \frac{\epsilon'_i}{2n}. $$
		Now by the Cauchy-Schwarz inequality, and having $|B| \leq n$, we have:
		$$ \sum_{x \in B} \left( \sum_{i=1}^s e_i(x) \right)^2 \ge \frac{1}{|B|} \left( \sum_{i=1}^s \sum_{x \in B} e_i(x) \right)^2 \ge \frac{ \left( \sum_{i=1}^s \epsilon'_i\right)^2}{4\,n}.$$
		\vspace{2mm}\noindent \textbf{Upper bound for the second term:\quad} On the other hand, for the second term in Equation~\eqref{eq:sum_error_terms}, we obtain:
		\begin{align*}
			\sum_{i=1}^s \sum_{x \in B} e_i(x)^2 & \leq \sum_{i=1}^s \sum_{x \in B} \frac{e_i(x)}{n} \leq \frac 1 {2n} \sum_{i=1}^s \epsilon'_i
		\end{align*}
		where the first inequality holds since the $e_i(x)$'s are at most $1/n$.

		\vspace{2mm}\noindent \textbf{Putting it all together:\quad} Using the two bounds above, we achieve the following lower bound for Equation~\eqref{eq:sum_error_terms}:
		$$ \sum_{ j < i } \sum_{x \in B} e_i(x)e_j(x)  \ge \frac{1}2 \left( \frac{\left(\sum_{i=1}^s \epsilon'_i\right)^2}{4n} - \frac{\sum_{i=1}^s \epsilon'_i}{2n} \right).$$
		Observe that since $s = c_1 \sqrt{n}/\epsilon^2$, for a sufficiently large $c_1$, $s$ is at least $\Theta(1/\epsilon)$. Therefore, we have:
		$$\sum_{i=1}^s \epsilon'_i \geq s \epsilon \geq 4\quad \Rightarrow \quad \frac 1 2\left(\sum_{i=1}^s \epsilon'_i \right)^2 - \sum_{i=1}^s \epsilon'_i \geq \frac{\left(\sum_{i=1}^s \epsilon'_i \right)^2 }{4}\,.$$
		Therefore, we obtain: 
		$$ \sum_{ j < i } \sum_{x \in B} e_i(x)e_j(x)  \ge \frac{1}2 \left( \frac{\left(\sum_{i=1}^s \epsilon'_i\right)^2}{4n} - \frac{\sum_{i=1}^s \epsilon'_i}{2n} \right) \geq \frac{\left(\sum_{i=1}^s \epsilon'_i \right)^2 }{16 n} \ge \frac{s^2 \epsilon^2}{16n}\,.$$ 
		Going back to Equation~\eqref{eq2}, we achieve:
		$$ \E \left[ \binom{s}2 Z \right] \ge \binom{s}2 \frac{1}n + \frac{s^2 \epsilon^2}{16n} \ge \binom{s}2 \frac{(1+\epsilon^2/8)}n\,,$$
		which concludes the proof of the lemma.
	\end{proof}
	
	We now proceed with the proof of Theorem \ref{thm:unif}. In the next step, we show a tight bound for the variance of the our statistic $Z$. We generalize the tight variance analysis given in \cite{DiakonikolasGPP16} for the standard collision based tester in the single source setting to our multiple source setting. We start by a useful identity for the variance: $\Var[Z] = \E[Z^2] - \E[Z]^2$. Note that when we expand $Z^2 = (\sum_{i<j} \sigma_{ij})^2$, we get terms of the form $\sigma_{ij}\sigma_{jk}$. In the single distribution case, this term can be related to the $\ell_3$-norm of $p$. In our setting, we introduce a generalization of the $\ell_3$-norm which is the following:
	\begin{equation} \label{eq:l3norm}
	\E[\sigma_{ij}\sigma_{jk}] = \sum_{x \in [n]} p_i(x)p_j(x)p_k(x) .
	\end{equation}
	To upper bound Equation\@ \eqref{eq:l3norm}, we again make use of the \textit{structural condition} and relate it to our version of the $\ell_2$-norm, Equation\@ \eqref{eq:l2norm}, by using Maclaurin's inequality which roughly states that the $\ell_3$-norm is at most the $\ell_2$-norm. More formally, we have the following lemma and our proof is presented in Section \ref{sec:proof_of_lemVarUB}.

	\begin{restatable}{lemma}{lemVarUB} \label{lem:var_ub}
		Let $\{p_i\}_{i=1}^s$ be distributions over $[n]$ that are all $\epsilon$-far from $\UU$ in $\ell_1$-distance and satisfy the \textit{structural condition} given in Definition \ref{def:structural} with $q = \UU_n$. Let $X_i$ be drawn independently from $p_i$ for all $1 \le i \le s$. Let $\sigma_{ij}$ be the indicator variable for the event $X_i = X_j$. Then the following estimate holds
		$$ \Var\left( \sum_{i < j} \sigma_{ij} \right) \le  \frac{18 \alpha s}{n^2}\binom{s}2 + 3 \left(\frac{\alpha}{n} \binom{s}2\right)^{3/2} + \sum_{i < j} \E[\sigma_{ij}]$$
		where $\alpha$ is defined to be the solution to $\E[Z] = (1 + \alpha)/n$, and it is at least $\epsilon^2/8$.
	\end{restatable}
	We can now prove the correctness of the algorithm by bounding the probability that $Z$ is below the threshold $\tau$. Recall that $\E[Z] = (1+\alpha)/n \geq (1+\epsilon^2/8)/n$ from Lemma~\ref{lem:ev_lb}. Therefore, $\alpha \geq \epsilon^2/8$. By Chebyshev's inequality, we have
	\begin{align*}
		\Pr\left[Z < \tau \right] & = \Pr\left[ \E[Z] - Z \geq \E[Z] - \tau \right] \leq  \Pr\left[ \left| \E[Z] - Z \right| \geq \E[Z] - \tau \right]
		\\ & 
		\leq \Pr\left[ \left| \E[Z] - Z \right| \geq  \frac{\alpha - \epsilon^2/16}{n}\right] 
		\leq \Var[Z] \cdot \left(\frac{n}{\alpha - \epsilon^2/16}\right)^2
		\\ & 
		\leq \Var\left[\sum_{i<j} \sigma_{ij}\right] \cdot \left(\frac{n}{\binom{s}2(\alpha - \epsilon^2/16)}\right)^2
		\,.
	\end{align*}
	Now, Lemma \ref{lem:var_ub} gives us 
	$$ \Var\left( \sum_{i < j} \sigma_{ij} \right) \le  \underbrace{\frac{18 \alpha s}{n^2}\binom{s}2 + 3 \left(\frac{\alpha}{n} \binom{s}2\right)^{3/2}}_{T_1} + \underbrace{\sum_{i < j} \E[\sigma_{ij}]}_{T_2}.$$
	We use $T_1$ and $T_2$ to indicate the two terms in the upper bound above. In either of the cases $T_1 \leq T_2$ or $T_1 > T_2$, we show the error probability is bounded by $1/3$. 
	\vspace{2mm}\noindent \textbf{Case 1: $\pmb{T_1 \leq T_2}$.\quad} In this case, we bound the variance of the number of collisions by $2\,T_2$. We have:
	\begin{align*}
		\Pr\left[Z < \tau \right] & \leq \Var\left[\sum_{i<j} \sigma_{ij}\right] \cdot \left(\frac{n}{\binom{s}2(\alpha - \epsilon^2/16)}\right)^2
		\leq 2 \,T_2 \cdot \left(\frac{n}{\binom{s}2(\alpha - \epsilon^2/16)}\right)^2
		\\ &
		\leq 2 \sum_{i<j} \E[\sigma_{ij}] \cdot \left(\frac{n}{\binom{s}2(\alpha - \epsilon^2/16)}\right)^2	 
		\leq \Theta\left(\frac{s^2 (1 + \alpha)}{n} \cdot \frac{n^2}{s^4 (\alpha - \epsilon^2/16)^2}\right)
		\\ & 
		\leq \Theta\left(\frac{n}{s^2} \cdot \underbrace{\frac{1 + \alpha}{(\alpha - \epsilon^2/16)^2}}_{f(\alpha)} \right) \,.
	\end{align*}
	Define $f(\alpha) := (1 + \alpha)/(\alpha - \epsilon^2/16)^2$. We can compute that $f$ is a decreasing function over the range $[\epsilon^2/8, \infty)$, so we can bound $f(\alpha)$ by $f(\epsilon^2/8) = \Theta(1/\epsilon^2)$ from above. Thus, we bound the probability of $Z < \tau$ as
	\begin{align*}
		\Pr\left[Z < \tau \right] & \leq \Theta\left(\frac{n}{s^2 \epsilon^2} \right) \leq \frac 1 3\,,
	\end{align*}
	where the last inequality holds for a sufficiently large constant $c_1$ and having $s = c_1 \sqrt{n}/\epsilon^2$.
	\vspace{2mm}\noindent \textbf{Case 2: $\pmb{T_1 > T_2}$.\quad} In this case, we bound the variance of $Z$ by $2 \,A$. Note that we know $\alpha \geq \epsilon^2/8$, so we have:
	\begingroup
	\allowdisplaybreaks
	\begin{align*}
		\Pr\left[Z < \tau \right] & \leq \Var\left[\sum_{i<j} \sigma_{ij}\right] \cdot \left(\frac{n}{\binom{s}2(\alpha - \epsilon^2/16)}\right)^2
		\leq 2 \,T_2 \cdot \left(\frac{n}{\binom{s}2(\alpha - \epsilon^2/16)}\right)^2
		\\ 
		& \leq 2 \left(\frac{18 \alpha s}{n^2}\binom{s}2 + 3 \left(\frac{\alpha}{n} \binom{s}2\right)^{3/2} \right)
		\cdot \left(\frac{n}{\binom{s}2(\alpha - \epsilon^2/16)}\right)^2
		\\ & 
		\leq \Theta\left( \left(\frac{\alpha s^3}{n^2} + \frac{\alpha^{3/2} s^3}{n^{3/2}}\right) \cdot \frac{n^2}{s^4 \alpha^2}\right)
		\leq \Theta \left(\frac{1}{s \, \alpha} + \frac{\sqrt n}{s\,\sqrt \alpha}\right).
	\end{align*}
	\endgroup
	The number of samples, $s$ is chosen to be
	
	$$s  = c_1 \cdot \frac{\sqrt n}{\epsilon^2} \geq \Theta \left(\frac{1}{\epsilon^2} + \frac{\sqrt n}{\epsilon}\right)
	\geq \Theta\left( \frac{1}{\alpha} + \frac{\sqrt{n}}{\sqrt{\alpha}}\right) \,,$$
	and therefore, by picking a sufficiently large constant $c_1$, we can bound the probability of outputting the incorrect answer in the soundness case by $1/3$.
\end{proof}

\subsection{Proof of Lemma~\ref{lem:var_ub}}

\label{sec:proof_of_lemVarUB}

\lemVarUB*
\begin{proof}
	For simplicity, let $W$ denote $\sum_{i<j} \sigma_{ij}$. We bound the variance of $W$ from above in the following steps.
	\begin{align*}
		& \Var[W]  = \E[W^2] - \E[W]^2 = \E\left[\left(\sum_{i<j} \sigma_{ij}\right)^2\right] - \left(\sum_{i<j} \E[\sigma_{ij}]\right)^2
		\\ & = 
		\E\left[ \sum_{\substack{i < j, k < \ell \\ \text{ all distinct}}}
		\sigma_{ij} \sigma_{k \ell} + 2 \sum_{i < j < \ell} \left( \sigma_{ij} \sigma_{ik} + \sigma_{ij} \sigma_{jk} + \sigma_{ik} \sigma_{jk}\right) + \sum_{i < j} \sigma_{ij}^2 \right]
		\\ & - 
		\sum_{\substack{i < j, k < \ell \\ \text{ all distinct}}}
		\E[\sigma_{ij}]\,\E[\sigma_{k \ell}] - 2 \sum_{i < j < \ell} \left( \E[\sigma_{ij}]\,\E[\sigma_{ik}] + \E[\sigma_{ij}]\,\E[\sigma_{jk}] + \E[\sigma_{ik}]\,\E[\sigma_{jk}]\right) \\
		&- \sum_{i < j}\E[\sigma_{ij}]^2 
		\,.
	\end{align*}
	Note that if $i, j, k, $ and $\ell$ are all distinct, then $\sigma_{ij}$ is independent from $\sigma_{k\ell}$. Thus, we have: 
	$$\E[\sigma_{ij} \sigma_{k\ell}] = \E[\sigma_{ij}]\,\E[\sigma_{k \ell}]\,.$$
	Moreover, we know that $\E[\sigma_{ij}] \ge 1/n$ for all $i<j$ from Lemma~\ref{lem:ev_lb}. Having $\sigma_{ij}^2 = \sigma_{ij}$, we continue bounding the variance as follows:
	\begin{align*}
		\Var[W] & = 
		2 \sum_{i < j < \ell} \left( \E[\sigma_{ij} \sigma_{ik}] + \E[\sigma_{ij} \sigma_{jk}] + \E[\sigma_{ik} \sigma_{jk}]\right) \\
		&+ \E\left[\sum_{i < j} \sigma_{ij} \right]
		- \binom{s}3\frac{6}{n^2} - \sum_{i < j}\E[\sigma_{ij}]^2 
		\,.
	\end{align*}
	For now, we focus on the first sum in the right hand side above. We bound this term via the error terms we defined in Equation~\eqref{eq:error_terms}. We note that
	\begin{align*}
		&\E[\sigma_{ij} \sigma_{ik}]  = \E[\sigma_{ij} \sigma_{jk}] = \E[\sigma_{ik} \sigma_{jk}] =  \sum_{x \in [n]} p_i(x)p_j(x)p_{k}(x) \\
		&=\sum_{x \in A} \left( \frac{1}n + e_i(x) \right)\left( \frac{1}n + e_j(x) \right)\left( \frac{1}n + e_{k}(x) \right) \\
		&+ \sum_{x \in B} \left( \frac{1}n - e_i(x) \right)\left( \frac{1}n - e_j(x) \right)\left( \frac{1}n - e_{k}(x) \right)\\ 
		&\le \frac{1}{n^2} \left(\underbrace{\sum_{x \in A} e_i(x) - \sum_{x \in B} e_i(x)}_{ = 0} + \underbrace{\sum_{x \in A} e_j(x) - \sum_{x \in B} e_j(x)}_{ = 0} + \underbrace{\sum_{x \in A} e_k(x) - \sum_{x \in B} e_k (x)}_{ = 0} \right)
		\\ & + \frac{1}{n}\left(\sum_{x \in [n]} e_i(x) e_j(x) + e_i(x) e_k(x) + e_j(x) e_k(x) \right) + \sum_{x \in [n]} e_i(x) e_j(x) e_{k}(x) +  \sum_{x\in[n]}\frac{1}{n^3}
	\end{align*}
	where the last inequality holds since all the $e_i(x)$'s are non-negative. Therefore, we can continue bounding the variance as follows:
	\begin{align*}
		\Var[W] & \leq \binom{s}3 \frac {6}{n^2} +  \frac{18 \, s}{n} \sum_{i<j}\sum_{x \in [n]} e_i(x)e_j(x) \\
		&+ 6 \sum_{i<j<k} \sum_{x \in [n]} e_i(x) e_j(x) e_{k}(x) + \E[W] - \binom{s}3\frac{6}{n^2} 
		\,.
	\end{align*}
	
	To bound the above terms, we use Maclaurin's inequality proved in~\cite{maclaurin}. 
	\begin{lemma}[Maclaurin's inequality] 
		\label{lem:maclaurin}
		Let $\{a_i\}_{i=1}^s$ be non-negative real numbers. Define
		$$ S_k = \frac{ 1}{\binom{s}{k}} \sum_{1 \le i_1 < \cdots < i_k \le s} a_{i_1}a_{i_2} \cdots a_{i_k} .$$
		Then,
		$$ S_1 \ge \sqrt{S_2} \ge \sqrt[3]{S_3} \ge \cdots \ge \sqrt[s]{S_s}.$$
	\end{lemma}
	For our purposes, we prove a strengthening of Maclaurin's inequality which is given in the following lemma:
	\begin{lemma}(Strengthened Maclaurin's Inequality) \label{lem:macuse}
		$$ \left( 2\sum_{\substack{ i < j < k }} \sum_{x \in [n]} e_i(x)e_j(x)e_{k}(x) \right)^2 \le \left( \sum_{i < j} \sum_{x \in [n]} e_i(x)e_j(x) \right)^3. $$ 
	\end{lemma}
	\begin{proof}
		We prove this by induction on $n$. Consider the case $n=1$.
		By Lemma \ref{lem:maclaurin}, we have: 
		$$ \left( \sum_{i < j} e_i(1)e_j(1) \right)^3  \ge \frac{ \binom{s}2^3}{\binom{s}3^2} \left(  \sum_{i < j < k} e_i(1)e_j(1)e_{k}(1) \right)^2.$$
		Now for $s \geq 3$, we have:  $\binom{s}2^3/\binom{s}3^2 > 4$  which proves our claim. We now proceed by induction. Suppose that the induction hypothesis is true for $n-1$. We know by the induction hypothesis that the following two inequalities hold:
		\begin{align*}
			\left(  \underbrace{  2\sum_{\substack{ i < j < k }} \sum_{x \in [n-1]} e_i(x)e_j(x)e_{k}(x) }_{F}\right)^2 &\le \left( \underbrace{\sum_{i < j} \sum_{x \in [n-1]} e_i(x)e_j(x)}_{F'} \right)^3 \\
			\left( \underbrace{ 2\sum_{\substack{ i < j < k }}  e_i(n)e_j(n)e_{k}(n) }_{G} \right)^2 &\le  \left( \underbrace{\sum_{i < j} e_i(n)e_j(n)}_{G'} \right)^3 
		\end{align*}
		where the first inequality is the induction hypothesis and the second inequality is just the base case of the induction which was proved earlier. Let $F, F', G,$ and $G'$ denote the terms as indicated above. The above inequalities after substituting new variables become: $F'^3 \ge F^2$ and $G'^3 \ge G^2$. Since all of these terms are positive, we have:
		$$(F'^3G'^3)^{1/2} \ge FG $$
		Then, by the arithmetic mean-geometric mean inequality, we have
		\begin{align*}
			3 F'^2 G' + 3F'G'^2 \geq 6 (F'G')^{3/2} \geq 6 FG \geq 2 FG.
		\end{align*}
		Using the fact that $F'^3 \ge F^2$ and $G'^3 \ge G^2$ again, yields
		$$ (F'+G')^3 \ge (F+G)^2\,,$$
		which concludes the lemma.
	\end{proof}
	We now proceed to bound the variance of $W$. We know
	\begin{align*}
		\Var[W] & \leq  \frac{18 \, s}{n} \sum_{i<j}\sum_{x \in [n]} e_i(x)e_j(x) + 6 \sum_{i<j<k} \sum_{x \in [n]} e_i(x) e_j(x) e_{k}(x) + \E[W] 
		\,.
	\end{align*}
	In the next step, we bound the two middle term based on $n$, $s$, and $\alpha$.
	Using Lemma \ref{lem:macuse}, we have
	$$\sum_{i < j < k} \sum_{x \in [n]} e_i(x) e_j(x) e_{k}(x) \le \frac{1}2 \left( \sum_{i < j} \sum_{x \in [n]} e_i(x)e_j(x) \right)^{3/2}\,.$$
	Recall that $\E[Z] = (1+\alpha)/n$. Thus, using Equation~\eqref{eq1}, we know
	$$\binom{s}2 \frac{1 + \alpha }{n} = \E[W] = \sum_{i<j} \sum_{x \in [n]} p_i(x)p_j(x) = \sum_{i<j} \sum_{x \in [n]} \left( e_i(x)e_j(x) + \frac{1}n\right) \,, $$
	which immediately implies that
	$$\sum_{i < j} \sum_{x \in [n]} e_i(x)e_j(x) = \binom{s}2\frac{\alpha}{n}\,.$$
	Putting all of it together, we obtain
	\begin{align*}
		\Var[W] & \leq  \frac{18 \, \alpha \, s }{n^2} \binom{s}2 + 3\left(\binom{s}2\frac{\alpha}{n}\right)^{3/2} +  \sum_{i < j} \E[\sigma_{ij}]
		\,,
	\end{align*}
	as desired.
\end{proof}

\section{Identity Testing with Multiple Sources} \label{sec:id_testing}
In this section, we present our algorithm for identity testing with multiple sources and its analysis. Recall that our goal is to distinguish the following two cases with probability at least $2/3$ given knowledge of some fixed distribution $q$ over $[n]$:
\begin{itemize}
	\item {\textbf{Completeness case:}} $p_1, p_2, \cdots, p_{s}$ are identical to $q$
	\item{\textbf{Soundness case:}} $p_1, p_2, \cdots, p_{s}$ are all $\epsilon$-far from $q$ in $\ell_1$-distance
\end{itemize}
where we receive samples from $\{p_i\}_{i=1}^s$. In the soundness case, we also assume that the $p_i$'s satisfy the  \textit{structural condition} given in Definition~\ref{def:structural}: we assume there are disjoint sets $A$ and $B$ that partition $[n]$ such that all $p_i$'s are larger than $q$ on the indices in $A$, and all the $p_i$'s are smaller than $q$ on the indices in $B$. Note that \textit{structural condition} trivially holds in the completeness case. 

\subsection{Algorithm for Identity Testing} \label{sec:id_alg}
We now present our algorithm, \textsc{Identity-Tester}, for identity testing with multiple sources. Suppose we receive $\text{Poi}(1)$ samples from each of the distributions $p_i$. This is a generalization of the standard technique in distribution testing which significantly simplifies the analysis of our algorithm by making certain random variables independent, as we explain later. Furthermore, as $\text{Poi}(s)$ is tightly concentrated around $s$, we can carry out this poissonization method at the expense of only constant factor increases in the sample complexity. Moreover, while we draw one sample in expectation per source, with probability 0.9, we will not receive more than $O(\log s)$ samples per distribution.

Our algorithm calculates a new $\chi$-square type statistics inspired by the previous $\chi^2$-type statistics~\cite{ValiantV14, ADK15, ChanDVV14, DiakonikolasK:2016}). The statistic is designed so that its expected value is related to the `$\ell_2$-norm' of the difference of the distributions, as explained in Section~\ref{sec:contributions}.  Similarly to uniformity testing, our algorithm in this section also proceeds by taking samples  and calculating our statistic. Then, it compares the value of this statistic to a threshold $\tau$. If the value of the statistic is `large', the algorithm outputs \reject and aborts, and outputs \accept otherwise. Ultimately, we prove that the sample complexity of our generalized identity tester depends on the $\ell_2$-norm of $q$, the known distribution. We give a flattening procedure in Section~\ref{sec:flattening_id} which allows us to assume that the $\ell_2$-norm of the known distribution is $O(1/\sqrt{n})$, resulting in the optimal sample complexity. We present our algorithm below along with the main theorem, Theorem~\ref{thm:id}, which proves the correctness of our algorithm. 

\begin{algorithm}[!ht]
	\SetKwInOut{Input}{Input}
	\SetKwInOut{Output}{Output}
	\Input{$n, \epsilon, q$, $\text{Poi}(1)$ samples from each of $p_1, p_2, \cdots, p_s$}
	\Output{\accept or \reject}
	\DontPrintSemicolon
	$s \gets \frac{c_1n \|q\|_2}{\epsilon^2}$ \;
	Draw $\text{Poi}(1)$ samples from each of the $s$ distributions $\{p_j\}_{j=1}^s$.\;
	$T_x \gets$ $\#$ times we see element $x \in [n]$ among the samples. \;
	$\tau \gets  \frac{5s^2 \epsilon^2}{8n}$ \;
	$Z \gets \sum_{x \in [n]} (T_x-sq(x))^2-T_x$\;
	\If{$Z \ge \tau$ }{Output \reject and abort.}
	Output \accept
	\caption{$\textsc{Identity-Tester}$}
	\label{alg:identity-tester}
\end{algorithm}

\begin{theorem}[Correctness of \textsc{Identity-Tester}]\label{thm:id}
	There exist a constant $c_1$ independent of $n$ such that the following statements hold with probability $2/3$:
	\begin{itemize}
		\item \textsc{Identity-Tester} outputs \accept if  each of the $s$ distributions $p_1, \cdots, p_s$ are equal to $q$.
		\item \textsc{Identity-Tester} outputs \reject if the $p_i$'s are $\epsilon$-far from $q$ ($\|p_i - q \|_1 \ge \epsilon$) and $\{p_i\}_{i=1}^s$ satisfy the \textit{structural condition} given in Definition~\ref{def:structural}.
	\end{itemize}
\end{theorem}
\begin{remark} \label{remark:identity_sample_complexity}
	The sample complexity of Algorithm~\ref{alg:identity-tester} is $\Theta(n\|q\|_2/\epsilon^2)$. Using the flattening procedure of Section~\ref{sec:flattening_id}, the sample complexity of \textsc{Identity-Tester} reduces to $\Theta(\sqrt{n}/\epsilon^2)$ which is optimal since the lower bound of $\Omega(\sqrt{n}/\epsilon^2)$ holds for identity testing in the standard single distribution setting~\cite{Paninski:08}.
\end{remark}

\begin{remark}
	Note that identity testing is a generalization of uniformity testing in Section~\ref{sec:unif_testing}. However, we keep our approach for uniformity testing since we only use \emph{exactly one} sample per distribution, rather than one sample in expectation.
\end{remark}

\vspace{2mm}\noindent \textbf{Overview of the proof:\quad} To prove the correctness of  $\textsc{Identity-Tester}$, we analyze the statistic 
$$ Z = \sum_{x \in [n]}(T_x-s\,q(x))^2- T_x \,,$$
where $T_x$ denotes the number of times we observe element $x$ among our $s' \sim  \text{Poi}(s)$ samples. Note that by employing the poissonization method,  $T_x$ is a Poisson random variable with parameter $\lambda_x \coloneqq \sum_{j=1}^s p_j(x)$ and $T_x$ and $T_y$ are independent for $x \ne y$. The independence among $T_x$'s greatly simplifies our calculations. 

Our goal is to show $Z$ is below the threshold $\tau$ in the completeness case, and above the threshold in the soundness case. To do so, we make use of the \textit{structural condition} to first define a convenient representation of $\E[Z]$ in Lemma~\ref{lem:id_ev}. We then show a strong concentration around its expectation by bounding the variance of $Z$ in Lemma~\ref{lem:id_var}. Finally, we show that $Z$ is always on the desired side of the threshold proving the correctness of our algorithm.
\begin{proof}[Proof of Theorem~\ref{thm:id}]
	Note that we set the (expected) number of samples to be $s = c_1n \|q\|_2/\epsilon^2$ for some sufficiently large constant $c_1$, and the threshold $\tau$ to be equal to $5s^2 \epsilon^2/(8n)$. We begin by stating a convenient representation of $\E[Z]$. To motivate our calculations, note that for a fixed $x$,
	\begin{align*}
		\E[(T_x-s\,q(x))^2- T_x] &= \E[T_x^2]-\E[T_x] - 2sq(x) \E[T_x] + s^2q(x)^2\\
		&= \lambda_x^2 -2s\,q(x)\lambda_x + s^2q(x)^2 = (\lambda_x - s\,q(x))^2
	\end{align*}
	which follows from the fact that $T_x$ is a Poisson random variable with parameter $\lambda_x = \sum_{j=1}^s p_j(x)$. Now using the \textit{structural condition}, we can define error terms similar to our uniformity testing section. For each distribution $p_j$, we define
	\begin{align*}
		e_j(x) &= p_j(x)-q(x) \quad \quad  \forall x \in A \\
		e_j(x)& = q(x) - p_j(x) \quad \quad  \forall x \in B.
	\end{align*}
	After plugging in $e_j(x)$ for all $x$ into our expression for $\lambda_x$, we combine these terms into a more useful representation of $\E[Z]$. We precisely show this representation  in Lemma~\ref{lem:id_ev} where we prove that $\E[Z]$ is given by  $\| \vec{\mathbf{e}}_1 + \cdots + \vec{\mathbf{e}}_s\|_2^2$
	where we interpret the vector $\vec{\mathbf{e}}_j \in \mathbb{R}^n$ as the vector with entries $e_j(x) = |q(x)-p_j(x)|$. Note that this is a natural generalization of the quantity $s^2\|q-p\|_2^2$ which is the quantity calculated by all $\chi^2$-based testers in the single distribution setting of identity testing (where all the samples are i.i.d.\@ from a fixed distribution $p$). More formally, we have the following lemma which we prove in Section~\ref{sec:id_ev}.
	
	\begin{restatable}{lemma}{lemidEV}\label{lem:id_ev}
		Let $\{p_i\}_{i=1}^s$ be distributions over $[n]$ that satisfy the \emph{structural condition} given in Definition~\ref{def:structural}. Suppose we draw $\textup{Poi}(1)$ samples from each $p_i$ and let $T_x$ be the number of times we see element $x \in [n]$ among the samples. Let $Z = \sum_{x \in [n]} (T_x-sq(x))^2-T_x$. Then,
		$$ \E[Z]= \| \vec{\mathbf{e}}_1 + \cdots + \vec{\mathbf{e}}_s\|_2^2$$
		where the $x$-th coordinate of $\vec{\mathbf{e}}_j \in \mathbb{R}^n$ is $|q(x)-p_j(x)|$.
	\end{restatable}
	We now give a tight upper bound for the variance of our statistic $Z$. Let $Z_x$ denote the $x$-th term in $Z$, $(T_x-sq(x))^2-T_x$. As we establish earlier, using the Poissonization method, $T_x$'s are independent from each other. Thus, the $Z_x$'s  are independent as well. Therefore, one can expand the variance of $Z$ as bellow: 
	$$\Var[Z] = \sum_{x \in [n]} \Var[Z_x]  = \sum_{x \in [n]} \E[Z_x^2] - \E[Z_x]^2\,. $$
	As we expand the term $Z_x^2$, to bound $\E\left[Z_x^2\right]$, higher norms of $T_x$, i.e.,  $\E\left[T_x^k\right]$ for $k \in [4]$, appear in our calculation. We can compute the closed-form of these quantities via the known norms of the Poisson distribution. Combining these terms, we again get an upper bound of $\Var[Z]$ in terms of the vectors $\vec{\mathbf{e}}_j$. Formally, we prove the following lemma in Section~\ref{sec:id_var}.
	\begin{restatable}{lemma}{lemidVAR} \label{lem:id_var}
		Let $\{p_i\}_{i=1}^s$ be $s$ distributions over $[n]$ that satisfy the structural condition given in Definition~\ref{def:structural}. Suppose we draw $\textup{Poi}(1)$ samples from each $p_i$, and let $T_x$ be the number of times we see element $x \in [n]$ among the samples. Let $Z = \sum_{x \in [n]} (T_x-sq(x))^2-T_x$. Then, we have:
		$$ \Var[Z] \le 4s \|q\|_2 \left \| \sum_{j=1}^s \vec{\mathbf{e}}_j \right \|_4^2 + 2 \left \| \sum_{j=1}^s \vec{\mathbf{p}}_j \right \|_2^2 + 4  \left \| \sum_{j=1}^s \vec{\mathbf{e}}_j \right \|_3^3\,,$$
		where $|q(x)-p_j(x)|$ is the $x$-th coordinates of $\vec{\mathbf{e}}_j \in \mathbb{R}^n$, and $\vec{p_j}$ is the vector representation of the distribution $p_j$.
	\end{restatable}
	We can now proceed to the proof of the theorem in the completeness case. 
	
	\vspace{2mm}\noindent \textbf{Proof of the completeness case:\quad} 
	In this case, Lemma~\ref{lem:id_ev} gives us $\E[Z] = 0$, and Lemma~\ref{lem:id_var} gives us $\Var[Z] \le 2s^2\|q\|_2^2$. Therefore by Chebyshev's inequality,
	$$\Pr[Z \ge \tau] \leq \Pr\left[|Z| = |Z - \E[Z]|  \ge \frac{s^2\epsilon^2}{4n} \right] \le \frac{32s^2\|q\|_2^2n^2}{s^4 \epsilon^4} = \frac{32\|q\|_2^2n^2}{s^2\epsilon^4}.$$
	Recall that we let $s = c_1n \|q\|_2/\epsilon^2$. The right hand side of the above inequality can be made arbitrarily small by picking a sufficiently large constant $c_1$, which proves the completeness case.
	
	\vspace{2mm}\noindent \textbf{Proof of the soundness case:\quad} In this case, Lemma~\ref{lem:id_ev} gives us
	$$ \E[Z] = \| \vec{\mathbf{e}}_1 + \cdots + \vec{\mathbf{e}}_s\|_2^2 = \sum_{x \in [n]} \left( \sum_{j=1}^s e_j(x) \right)^2 \ge \frac{1}n \left( \sum_{j=1}^s \sum_{x \in [n]} e_j(x) \right)^2 \ge \frac{s^2\epsilon^2}n $$
	where the first inequality is Cauchy-Schwarz, and the second inequality follows from the fact that $$ \sum_{x \in [n]} e_j(x) = \|q-p_j\|_1 \geq \epsilon$$ for each $j \in [s]$. Then by Chebyshev's inequality and Lemma~\ref{lem:id_var},
	\begin{align*}
		\Pr\left[ |Z - \E[Z] | \ge \frac{\E[Z]}{4} \right] &\le \frac{16 \Var[Z]}{\E[Z]^2} \\
		& \le  \frac{4s \|q\|_2\left \| \sum_{j=1}^s \vec{\mathbf{e}}_j \right \|_4^2}{\left \| \sum_{j=1}^s \vec{\mathbf{e}}_j \right\|_2^4} +   \frac{2 \left \| \sum_{j=1}^s \vec{\mathbf{p}}_j \right \|_2^2}{\left \| \sum_{j=1}^s \vec{\mathbf{e}}_j \right\|_2^4} +  \frac{4  \left \| \sum_{j=1}^s \vec{\mathbf{e}}_j \right \|_3^3}{\left \| \sum_{j=1}^s \vec{\mathbf{e}}_j \right\|_2^4} \numberthis \label{eq:chebyshevbound}
	\end{align*}
	Now, we bound each of the three terms above separately. We start off by introducing a new distribution denoted by $ \widetilde{\mathbf{p}}$ to be
	$ \widetilde{\mathbf{p}} := \frac{1}s \sum_{j=1}^s \mathbf{p}_j $. In some of our calculation, this new representation simplifies our calculations. 
	\begin{itemize}
		\item \textbf{First term:} Now, we focus on the first term in Equation~\eqref{eq:chebyshevbound}. We note that all the $e_j(x)$ are positive, so we have:
		\begin{equation}\label{eq:ell2-ell4}
		\begin{split}
			\left \| \sum_{j=1}^s \vec{\mathbf{e}}_j \right\|_4^2 & = \sqrt{\sum_{x\in[n]}\left(\sum_{j=1}^s  e_j(x) \right)^4} \leq \sum_{x\in[n]}\left(\sum_{j=1}^s  e_j(x) \right)^2  \le \left  \| \sum_{j=1}^s \vec{\mathbf{e}}_j \right\|_2^2\,. 
		\end{split}
		\end{equation}
		We again use the same Cauchy-Schwarz calculation as in $\E[Z]$ and get:
		\begin{equation}\label{eq:lb_ell2^2}
		    \left \| \sum_{j=1}^s \vec{\mathbf{e}}_j \right\|_2^2 = \sum_{x \in [n]}  \left( \sum_{j=1}^s e_j(x) \right)^2 \ge \frac{1}n \left(  \sum_{x \in [n]} \sum_{j=1}^s e_j(x) \right)^2 \ge \frac{s^2\epsilon^2}n\,.
		\end{equation}
		Therefore, we bound the first term from above:
		$$ \frac{4s \|q\|_2\left \| \sum_{j=1}^s \vec{\mathbf{e}}_j \right \|_4^2}{\left \| \sum_{j=1}^s \vec{\mathbf{e}}_j \right\|_2^4} \le \frac{4n \|q\|_2}{s \epsilon^2}. $$
		
		\item \textbf{Second term}: For the second term, we have:
		$$   \frac{2 \left \| \sum_{j=1}^s \vec{\mathbf{p}}_j \right \|_2^2}{\left \| \sum_{j=1}^s \vec{\mathbf{e}}_j \right\|_2^4}  = \frac{2s^2 \| \widetilde{\mathbf{p}} \|_2^2}{\left \| \sum_{j=1}^s \vec{\mathbf{e}}_j \right\|_2^4} $$
		where  $ \widetilde{\mathbf{p}} = \frac{1}s \sum_{j=1}^s \mathbf{p}_j $. We now consider two cases. If it is the case that 
		$ \| \widetilde{\mathbf{p}} \|_2 \le 3 \|q \|_2  $,
		then we have:
		$$\frac{2s^2 \| \widetilde{\mathbf{p}} \|_2^2}{\left \| \sum_{j=1}^s \vec{\mathbf{e}}_j \right\|_2^4} \le \left( \frac{5 n \|q\|_2}{s \epsilon^2} \right)^2.$$
		On the other hand, suppose that $ \| \widetilde{\mathbf{p}} \|_2 > 3 \|q \|_2 $. Then, using our \textit{structural condition}, we obtain:
		$$  \frac{2s^2 \| \widetilde{\mathbf{p}} \|_2^2}{\left \| \sum_{j=1}^s \vec{\mathbf{e}}_j \right\|_2^4} = \frac{2 \| \widetilde{\mathbf{p}} \|_2^2}{s^2\| \widetilde{\mathbf{p}}-q \|_2^4} $$
		and note that
		$$ \|\widetilde{\mathbf{p}} - q \|_2^2 \ge \| \widetilde{\mathbf{p}} \|_2^2 + \|q\|_2^2 - 2\|q\|_2 \|  \widetilde{\mathbf{p}}\|_2 \ge \| \widetilde{\mathbf{p}}\|_2^2/3 .$$
		Hence,
		$$  \frac{2 \| \widetilde{\mathbf{p}} \|_2^2}{s^2\| \widetilde{\mathbf{p}}-q \|_2^4} \le \frac{18}{s^2\|\widetilde{\mathbf{p}} \|_2^2} \le \left( \frac{5}{s \|q\|_2} \right)^2. $$
		
		\item \textbf{Third term}: Again, we use the Cauchy-Schwarz inequality to obtain:
		\begin{align*}
		    \left \| \sum_{j=1}^s \vec{\mathbf{e}}_j \right\|_3^3 & = \sum_{x\in[n]}\left(\sum_{j=1}^s  e_j(x) \right)^3 \leq \sqrt{\left(\sum_{x\in[n]}\left(\sum_{j=1}^s  e_j(x) \right)^2\right)\cdot \left(\sum_{x\in[n]}\left(\sum_{j=1}^s  e_j(x) \right)^4\right)}
		    \\& \leq \left \| \sum_{j=1}^s \vec{\mathbf{e}}_j \right\|_2 \cdot \left \| \sum_{j=1}^s \vec{\mathbf{e}}_j \right\|_4^2
		\end{align*}
		Now, we use Equation~\eqref{eq:ell2-ell4} and Equation~\eqref{eq:lb_ell2^2}, which we show earlier, to bound the third term:
		\begin{align*}
		    \frac{4  \left \| \sum_{j=1}^s \vec{\mathbf{e}}_j \right \|_3^3}{\left \| \sum_{j=1}^s \vec{\mathbf{e}}_j \right\|_2^4} 
		        & \leq \frac{4\,\left \| \sum_{j=1}^s \vec{\mathbf{e}}_j \right\|_2 \cdot \left \| \sum_{j=1}^s \vec{\mathbf{e}}_j \right\|_4^2}{\left \| \sum_{j=1}^s \vec{\mathbf{e}}_j \right\|_2^4} \leq \frac{4}{\left \| \sum_{j=1}^s \vec{\mathbf{e}}_j \right\|_2} \leq \frac{4\,\sqrt n}{s \epsilon}\,.
		\end{align*}
		
	\end{itemize}
	Thus in all cases, we have
	$$ \Pr\left[ |Z - \E[Z] | \ge \frac{\E[Z]}{4} \right] \le \frac{4n \|q\|_2}{s \epsilon^2} + \left( \frac{5n \|q\|_2}{s \epsilon^2} \right)^2  + \left( \frac{5}{s \|q\|_2} \right)^2 + \frac{4\,\sqrt n}{s \epsilon}.$$
	Note that $s = c_1n \|q\|_2/\epsilon^2$ and $\|q\|_2 \ge 1/\sqrt{n}$. Therefore, by letting $c_1$ be a sufficiently large constant,  we get that the above probability is smaller than $1/3$. Hence with probability at least $2/3$, we know $Z \ge 3s^2\epsilon^2/(4n)$ in the soundness case, so we reject with probability at least $2/3$, as desired.
\end{proof}

\subsection{Flattening Procedure} \label{sec:flattening_id}
In this section, we present the flattening procedure which allows us to assume that $\|q\|_2^2 = O(1/n)$ in Remark~\ref{remark:identity_sample_complexity} without loss of generality where $q$ is our known distribution. While this procedure is similar to the one used in \cite{DiakonikolasK:2016}, we state it here for the sake of completeness. For each $x \in [n]$, define
$$b_x \coloneqq \lfloor nq(x) \rfloor + 1. $$
We note that $b_x \ge 1$ for each $x \in [n]$.
Given a sample $x$ from a distribution $p$ over $[n]$, we can get a sample from the `flattened' distribution $p'$ over a new domain, $\mathcal{D}$, defined as
$$ \mathcal{D} \coloneqq \{(x,y) \mid x \in [n], y \in [b_x] \},$$
by drawing an element from $y \in [b_x]$ uniformly at random and creating the tuple $(x,y)$. This is the flattening procedure that we use for our version of identity testing. Note that the probability mass over $[n]$ placed by $p$ gets `flattened' to be a probability distribution over the domain $\mathcal{D}$.

Furthermore, this procedure has a few desired properties which we state here:
First, the size of this new domain is $O(n)$: 
$$ |\mathcal{D}| = \sum_{x\in[n]} b_x \le 2n. $$

Second, the procedure  preserves the $\ell_1$-distance between two distributions: let $p'$ and $q'$ denote the flattened versions of $p$ and $q$. Then, we have:
$$ \|q' - p'\|_1 = \sum_{x \in [n]} \sum_{y \in [b_x]} \frac{|q(x)-p(x)|}{|b_x|}  = \sum_{x \in [n]} |q(x)-p(x)| = \|q-p\|_1.$$
Third, by definition of the $b_x$'s, we can show that $q'$ has a low $\ell_2$-norm:
$$  \|q'\|_2^2 = \sum_{x \in [n]} \sum_{y \in [b_x]} \frac{q(x)^2}{b_x^2} = \sum_{x \in [n]} \frac{q(x)^2}{b_x} \le \sum_{x \in [n]} \frac{q(x)}n \le \frac{1}n \,.$$
The above inequality implies that the $\ell_2$-norm of $q'$ is within a constant factor of the smallest possible norm, which is $(|\mathcal{D}|)^{-1/2}$. 

Therefore whenever we get a sample over $[n]$ in \textsc{Identity-Tester}, we can use this flattening procedure to draw a sample over $\mathcal{D}$. By using this flattening procedure to draw samples from a slightly larger domain, we can assume that the $\ell_2$-norm of the known distribution $q$ is $O(1/\sqrt{n})$. Note that since the size of the larger domain is still $O(n)$, the flattening procedure only affects the sample complexity up to constant factors. Therefore, by combining with Theorem~\ref{thm:id}, we can perform our generalized version of identity testing by using $s = O(\sqrt{n}/\epsilon^2)$ samples, which is optimal up to constant factors.

\subsection{Proof of Lemma \ref{lem:id_ev}} \label{sec:id_ev}
\lemidEV*
\begin{proof}
	Let 
	$$Z_x =  (T_x-sq(x))^2-T_x.$$
	We have:
	$$Z_x^2 =  T_x^2 - 2sq(x)T_x + s^2q(x)^2 - T_x. $$
	We can compute that:
	\begin{align*}
		\E[Z_x] &= \E[T_x^2]-\E[T_x] - 2sq(x) \E[T_x] + s^2q(x)^2  \\
		&= \lambda_x^2 -2sq(x)\lambda_x + s^2q(x)^2
	\end{align*}
	where we have used the fact that the variance of a Poisson random variable with parameter $\lambda$ is also $\lambda$. We introduce the following notation $(-1)^{x \in B}$ which is defined as
	$$(-1)^{x\in B} = 
	\begin{cases}
	-1 \ &\text{if} \ x \in B, \\
	\ \ 1 \ &\text{if} \ x \not \in B.
	\end{cases}
	$$ Using the fact that $\lambda_x = \sum_{j=1}^sp_j(x)$, we can compute that
	$$\sum_{x \in [n]}\lambda_x= \sum_{j=1}^s \sum_{x \in [n]} \left(q(x)+(-1)^{x\in B}e_j(x)\right).$$
	In Appendix \ref{sec:moment_calcs}, we calculate the $\sum_{x \in [n]} \lambda_x^2$. There we show that
	\begin{align*}
		\sum_{x \in [n]} \lambda_x^2 = \  &s^2||q||_2^2 + 2s \sum_{j=1}^s \sum_{x \in [n]}(-1)^{x \in B}q(x)e_j(x) \\ 
		&+ \sum_{j=1}^s \sum_{x \in [n]} e_j(x)^2 + \sum_{j \ne k} \sum_{x \in [n]}e_j(x)e_k(x).
	\end{align*}
	Using these two results, we have:
	\begin{align*}
		\sum_{x \in [n]} \E[Z_x] &= \sum_{x \in [n]} \lambda_x^2 - 2s \sum_{x \in [n]} q(x) \lambda_x + s^2 \sum_{x \in [n]} q(x)^2 \\
		&= 2s \sum_{j=1}^s \sum_{x \in [n]} (-1)^{x \in B}q(x)e_j(x) + \sum_{j=1}^s \sum_{x \in [n]} e_j(x)^2 + \sum_{j\ne k} \sum_{ x \in [n]}e_j(x)e_k(x) \\
		&- 2s \sum_{j=1}^s \sum_{x \in [n]} (q(x)^2 + (-1)^{x \in B}q(x)e_j(x)) + 2s^2\|q\|_2^2 \\
		&= \sum_{j=1}^s \sum_{x \in [n]} e_j(x)^2 + \sum_{j\ne k} \sum_{x \in [n]} e_j(x)e_k(x) \\
		&= \| \vec{\mathbf{e}}_1 + \cdots + \vec{\mathbf{e}}_s\|_2^2.
	\end{align*}
	Therefore, our final result is as follows:
	$$ \E[Z] =  \| \vec{\mathbf{e}}_1 + \cdots + \vec{\mathbf{e}}_s\|_2^2,$$
	as desired.
\end{proof}
\begin{remark}
	Note that the quantity $\| \vec{\mathbf{e}}_1 + \cdots + \vec{\mathbf{e}}_s\|_2^2$ is a natural generalization of the quantity $s^2\|q-p\|_2^2$ which is the expectation of the random variable calculated by identity testers in the single distribution case where samples come from a fixed distribution $p$.
\end{remark}

\subsection{Proof of Lemma \ref{lem:id_var}} \label{sec:id_var}
\lemidVAR*
\begin{proof}
	Let
	$$Z_x = (T_x-sq(x))^2-T_x.$$
	Due to the independence of $T_x$, we have
	$$\Var[Z] = \sum_{x \in [n]} \Var[Z_x]  = \sum_{x \in [n]} \E[Z_x^2] - \E[Z_x]^2 $$
	where $$Z_x =  (T_x-sq(x))^2-T_x.$$
	Noting that $T_x$ is a Poisson with parameter $\lambda_x$, we can compute that
	$$ \E[Z_x^2] = \lambda_x^4 + 4\lambda_x^3(1-sq(x)) + 2 \lambda_x^2(1-4sq(x)+3s^2q(x)^2) + 4s^2q(x)^2 \lambda_x(1-sq(x)) + s^4q(x)^4.$$
	In Appendix \ref{sec:moment_calcs} we calculate the $\sum_{x \in [n]} \lambda_x^k$ for $k \in \{1,2,3,4\}$. Using these results and simplifying, we arrive at the following expression:
	\begin{align*}
		&\sum_{x \in [n]} \E[Z_x^2] = 4s \sum_{j = 1}^s \sum_{x \in [n]} q(x)e_j(x)^2 + 4s \sum_{j\ne k} \sum_{x \in [n]} q(x)e_j(x)e_k(x) +2s^2\|q\|_2^2 \\
		&+ 4s \sum_{j=1}^s \sum_{x \in [n]}(-1)^{x \in B}q(x)e_j(x) + 2 \sum_{j=1}^s \sum_{x \in [n]} e_j(x)^2 + 2 \sum_{j \ne k} \sum_{x \in [n]} e_j(x)e_k(x) \\
		&+ 4 \sum_{j=1}^s \sum_{x \in [n]} (-1)^{x \in B}e_j(x)^3 + 12 \sum_{j \ne k}  \sum_{x \in [n]} (-1)^{x \in B}e_j(x)^2e_k(x)  \\
		&+ 4\sum_{j \ne k \ne \ell} \sum_{x \in [n]} (-1)^{x \in B} e_j(x)e_k(x)e_{\ell}(x) + \Bigg ( \sum_{j=1}^s \sum_{x \in [n]} e_j(x)^4 + 6\sum_{j \ne k \ne \ell}e_j(x)^2e_k(x)e_\ell(x) \\
		&+ 4\sum_{j \ne k}e_j(x)^3e_k(x) + 3\sum_{j \ne k}e_j(x)^2e_k(x)^2  + \sum_{j \ne k \ne \ell \ne t} e_j(x)e_k(x)e_\ell(x)e_t(x) \Bigg ).
	\end{align*}
	We now simplify the above expression. First note that
	\begin{align*}
		&4s \sum_{j = 1}^s \sum_{x \in [n]} q(x)e_j(x)^2 + 4s \sum_{j\ne k} \sum_{x \in [n]} q(x)e_j(x)e_k(x) \\
		&= 4s \sum_{x \in [n]} q(x)\left( \sum_{j=1}^s e_j(x)^2 +  \sum_{j \ne k} e_j(x)e_k(x) \right) = 4s \sum_{x \in [n]} q(x) \left( \sum_{j=1}^s e_j(x)\right)^2 \\
		&\le 4s \|q\|_2 \| \vec{\mathbf{e}}_1 + \cdots + \vec{\mathbf{e}}_s \|_4^2.
	\end{align*}
	Furthermore, 
	\begin{align*}
		&2s^2\|q\|_2^2 + 4s \sum_{j=1}^s \sum_{x \in [n]}(-1)^{x \in B}q(x)e_j(x) + 2 \sum_{j=1}^s \sum_{x \in [n]} e_j(x)^2 + 2 \sum_{j \ne k} \sum_{x \in [n]} e_j(x)e_k(x) \\
		&= 2 \sum_{j \ne k} \sum_{x \in [n]} (q(x)+(-1)^{x \in B}e_j(x))(q(x)+(-1)^{x \in B}e_k(x)) \\
		&= 2 \| \vec{\mathbf{p}}_1 + \cdots \vec{\mathbf{p}}_s\|_2^2
	\end{align*}
	and
	\begin{align*}
		&4 \sum_{j=1}^s \sum_{x \in [n]} (-1)^{x \in B}e_j(x)^3 + 12 \sum_{j \ne k}  \sum_{x \in [n]} (-1)^{x \in B}e_j(x)^2e_k(x)  \\
		&+ 4\sum_{j \ne k \ne \ell} \sum_{x \in [n]} (-1)^{x \in B} e_j(x)e_k(x)e_{\ell}(x)\\
		&\le 4\left( \sum_{j=1}^s \sum_{x \in [n]}e_j(x)^3 + 3 \sum_{j \ne k}  \sum_{x \in [n]} e_j(x)^2e_k(x) + \sum_{j \ne k \ne \ell} \sum_{x \in [n]} e_j(x)e_k(x)e_{\ell}(x) \right) \\
		&= 4 \|\vec{\mathbf{e}}_1 + \cdots + \vec{\mathbf{e}}_s \|_3^3.
	\end{align*}
	Finally, the last expression inside the parenthesis in the expression for $\sum_{x\in[n]} \E[Z_x^2]$ which is:
	\begin{align*}
	   &\sum_{j=1}^s \sum_{x \in [n]} e_j(x)^4 + 6\sum_{j \ne k \ne \ell}e_j(x)^2e_k(x)e_\ell(x) + 4\sum_{j \ne k}e_j(x)^3e_k(x) + 3\sum_{j \ne k}e_j(x)^2e_k(x)^2  \\
	   &+ \sum_{j \ne k \ne \ell \ne t} e_j(x)e_k(x)e_\ell(x)e_t(x) 
	\end{align*}
	is precisely $\sum_{x \in [n]} \left( \sum_{j=1}^s e_j(x) \right)^4$. Therefore, we obtain the following bound:
	$$ \sum_{x \in [n]} \E[Z_x^2] \le 4s \|q\|_2 \left \| \sum_{j=1}^s \vec{\mathbf{e}}_j \right \|_4^2 + 2 \left \| \sum_{j=1}^s \vec{\mathbf{p}}_j \right \|_2^2 + 4  \left \| \sum_{j=1}^s \vec{\mathbf{e}}_j \right \|_3^3 + \sum_{x \in [n]} \left( \sum_{j=1}^s e_j(x) \right)^4.$$ 
	Using the calculation for $\E[Z_x]$ that we performed for Lemma \ref{lem:id_ev}, we see that:
	$$ \sum_{x \in [n]} \E[Z_x]^2 = \sum_{x \in [n]} \left( \sum_{j=1}^s e_j(x)^2 + \sum_{j \ne k} e_j(x)e_k(x) \right)^2 = \sum_{x \in [n]} \left( \sum_{j=1}^s e_j(x) \right)^4 $$
	so altogether,
	$$
	\Var[Z] \le 4s \|q\|_2 \left \| \sum_{j=1}^s \vec{\mathbf{e}}_j \right \|_4^2 + 2 \left \| \sum_{j=1}^s \vec{\mathbf{p}}_j \right \|_2^2 + 4  \left \| \sum_{j=1}^s \vec{\mathbf{e}}_j \right \|_3^3,
	$$
	as desired.
\end{proof}
\begin{remark}
	As stated in Section \ref{sec:id_ev}, the quantity $\left \| \sum_{j=1}^s \vec{\mathbf{e}}_j \right \|_4^2$ is a natural generalization of $s^2\|q-p\|_2^2$ which appears in the variance calculations of the statistic used by identity testers in the single distribution setting where samples come from a fixed distribution $p$. Similarly, the quantity $\left \| \sum_{j=1}^s \vec{\mathbf{p}}_j \right \|_2^2$ is a natural generalization of $s^2\|p\|_2^2$ which also appears in these calculations.
\end{remark}

\section{Closeness Testing with Multiple Sources} \label{sec:closeness_testing}
in this section, we present our algorithm for closeness testing with multiple sources and its analysis. Recall that our goal is to distinguish the following two cases with probability at least $2/3$:
\begin{itemize}
	\item {\textbf{Completeness case:}} $p_1, p_2, \cdots$ are identical to $q$
	\item{\textbf{Soundness case:}} $p_1, p_2, \cdots$ are all $\epsilon$-far from $q$ in $\ell_1$-distance
\end{itemize}
where we have access to two streams. In the first stream, all samples are i.i.d.\@ from some fixed distribution $q$ over $[n]$ that is unknown to us while in the second stream, we receive samples from $\{p_i\}_{i=1}^s$. In the soundness case, we also assume that the $p_i$'s satisfy the \textit{structural condition} given in Definition \ref{def:structural}: we assume there are disjoint sets $A$ and $B$ that partition $[n]$ such that all $p_i$'s are larger than $q$ on the indices in $A$, and all the $p_i$'s are smaller than $q$ on the indices in $B$. Note that the \textit{structural condition} is trivially satisfied in the completeness case.

\subsection{Algorithm for Closeness Testing}\label{sec:closeness_alg}
We now present our algorithm, \textsc{Closeness-Tester}, for closeness testing with multiple sources. Overall, our approach is very similar to our approach in Section \ref{sec:id_testing}. As in the identity testing section, we again make the assumption that we are able to receive $\text{Poi}(1)$ samples from each of the distributions $p_i$. This is a generalization of the standard assumption in the single source closeness testing~\cite{ChanDVV14,DiakonikolasK:2016} and it significantly simplifies the analysis of our algorithm by making certain random variables independent, similar to the identity testing section. Furthermore, as $\text{Poi}(s)$ is tightly concentrated around $s$, we can carry out this poissonization method at the expense of only constant factor increases in the sample complexity.

As in the identity testing section, the statistic calculated by our algorithm is introduced in~\cite{ChanDVV14, DiakonikolasK:2016}. We show that the expected value of our statistic is related to the `$\ell_2$-norm' of the difference of the distributions, as explained in Section \ref{sec:contributions}. If the value of the statistic is `large' compared to some threshold $\tau$, the algorithm outputs \reject and outputs \accept otherwise. As in uniformity testing, the challenge is to show that the value of this statistic concentrates which we do so by analyzing its variance. Ultimately, we prove that the sample complexity of our generalized identity tester depends on the $\ell_2$-norm of $q$, the distribution that we have i.i.d.\@ sample access from. We then present a randomized flattening procedure in Section \ref{sec:flattening_closeness} which shows how to reduce the $\ell_2$-norm of $q$ to $O(1/\sqrt{k})$ where $k$ is a parameter in our algorithm. This randomized flattening procedure is slightly different than the one used in Section \ref{sec:flattening_id} since we do not know $q$ in advance. Therefore, we must use some samples from $q$ to towards this procedure. We present our algorithm below along with the main theorem, Theorem \ref{thm:closeness}, which proves the correctness of our algorithm. 

\begin{algorithm}[!ht]
	\SetKwInOut{Input}{Input}
	\SetKwInOut{Output}{Output}
	\Input {$n, \epsilon$, sample access to $q$, \text{Poi}(1) samples from each of $p_1, p_2, \cdots, p_s$}
	\Output{Accept or Reject}
	\DontPrintSemicolon
	$k \gets \frac{n^{2/3}}{\epsilon^{4/3}}$ \;
	Draw $\text{Poi}(k)$ samples from $q$ to perform the randomized flattening procedure given in Section \ref{sec:flattening_closeness}. \;
	$s \gets \frac{c_1n}{\epsilon^2\sqrt{k}}$ \;
	
	Draw $\text{Poi}(s)$ samples from $q$. \;
	Draw $\text{Poi}(1)$ samples from each of the $s$ distributions $\{p_j\}_{j=1}^s$. \;
	$Y_x \gets$ $\#$ times we see element $x \in [n]$ among the $\text{Poi}(s)$ samples from $q$ \;
	$T_x \gets$ $\#$ times we see element $x \in [n]$ among the samples from $\{p_j\}_{j=1}^s$. \;
	$\tau \gets \frac{5s^2 \epsilon^2}{8n}$ \;
	$Z \gets \sum_{x \in [n]} (T_x-Y_x)^2-T_x-Y_x$\;
	\If{$Z \ge \tau$ }{Output \reject and abort.}
	Output \accept
	\caption{$\textsc{Closeness-Tester}$}
	\label{alg:closeness-tester}
\end{algorithm}

\begin{theorem}[Correctness of \textsc{Closeness-Tester}]\label{thm:closeness}
	There exist a constant $c_1$ independent of $n$ such that the following statements hold with probability $2/3$:
	\begin{itemize}
		\item \textsc{Closeness-Tester} outputs \accept if each of the $s$ distributions $p_1, \cdots, p_s$ are $q$.
		\item \textsc{Closeness-Tester} outputs \reject if each of the $p_i$'s are $\epsilon$-far from $q$ ( $\|p_i-q\|_1 \ge \epsilon$) and $\{p_i\}_{i=1}^s$ satisfy the structural condition given in Definition \ref{def:structural}.
	\end{itemize}
\end{theorem}
\begin{remark}
	The sample complexity of $\textsc{Closeness-Tester}$ is $\Theta(k + n/(\epsilon^2 \sqrt{k}))$ using the flattening procedure of Section~\ref{sec:flattening_closeness}. Optimizing in $k$, the sample complexity of \textbf{Closeness-Tester} reduces to $\Theta(n^{2/3}/\epsilon^{4/3} + \sqrt{n}/\epsilon^2)$ which is optimal since the same lower bound holds for closeness testing in the standard single distribution setting~\cite{DiakonikolasK:2016}.
\end{remark}

\begin{remark}\label{remark:unequal-sized-closeness}
    Note that given the above result, we may need fewer sources as long as we have more samples from the distribution $q$. In particular, suppose we use $k_1 = \Theta\left(\min\left(n^{2/3}/\epsilon^{4/3} + \sqrt{n}/\epsilon^2, n\right)\right)$ samples from $q$ for flattening, which implies that the $\ell_2$-norm of the  ``flattened'' $q$ is $O(1/\sqrt{k_1})$ with high probability. 
    Then, we can use $s = \Theta(n/(\sqrt{k_1} \, \epsilon^2)+ \sqrt{n}/\epsilon^2)$ sources, and $\Theta(s)$ samples from the ``flattened'' $q$ and use  \textbf{Closeness-Tester} to distinguish between the completeness and the soundness case. This is an sample-optimal trade off up to constant factors since the same lower bound holds in the standard single distribution setting~\cite{DiakonikolasK:2016}.
\end{remark}

\vspace{2mm}\noindent \textbf{Overview of the proof:\quad} 
To prove the correctness of $\textsc{Closeness-Tester}$, we analyze the statistic 
$$ Z = \sum_{x \in [n]}(T_x-Y_x)^2- T_x-Y_x $$
where $T_x$ denote the number of times we observe element $x$ among the samples from $\{p_i\}_{i=1}^s$ and $Y_x$ denotes the number of times we see element $x$ among the $\text{Poi}(s)$ samples from $q$. Note that by the poissonization method, $T_x$ is a Poisson random variable with parameter $\lambda_x = \sum_{j=1}^s p_j(x)$, similar to the identity testing section. Furthermore, $T_x$ and $T_y$ are independent for $x \ne y$ and furthermore, $Y_x$ is a Poisson random variable with parameter $s\,q(x)$. Our goal is to show $Z$ is below the threshold $\tau$ in the completeness case, and above the threshold in the soundness case. To do so, we make use of the \textit{structural condition} to first define a convenient representation of $\E[Z]$ in Lemma \ref{lem:closeness_ev}. We then show a strong concentration around its expectation by bounding the variance of $Z$ in Lemma \ref{lem:closeness_var}. These two lemmas are very similar to the ones proved in Section \ref{sec:id_testing} due to the similarity of the statistic $Z$ used here and in Section \ref{sec:id_testing}. Finally, we show that $Z$ is always on the desired side of the threshold proving the correctness of our algorithm.

\begin{proof}[Proof of Theorem \ref{thm:closeness}] Note that we set the number of samples to be $c_1n/(\epsilon^2 \sqrt{k})$ for a sufficiently large constant $c_1$, and the threshold $\tau$ to be equal to $5s^2\epsilon^2/(8n)$. We begin by stating a convenient representation of $\E[Z]$. To motivate our calculations, note that for a fixed $x$, 
	\begin{align*}
		\E[(T_x-Y_x)^2-T_x-Y_x] &= \E[T_x^2]-\E[T_x] + \E[Y_x^2]-\E[Y_x]-2\E[Y_x]\E[T_x]  \\
		&= \lambda_x^2 -2sq(x)\lambda_x + s^2q(x)^2
	\end{align*}
	where we have used the fact that $T_x$ and $Y_x$ are Poisson random variables with parameters $\lambda_x = \sum_{j=1}^s p_j(x)$ and $sq(x)$ respectively. Now using the \textit{structural condition}, we can define error terms similar to Sections \ref{sec:unif_testing} and \ref{sec:id_testing}. For each distribution $p_j$, we define
	\begin{align*}
		e_j(x) &= p_j(x)-q(x) \quad \quad  \forall x \in A \\
		e_j(x)& = q(x) - p_j(x) \quad \quad  \forall x \in B.
	\end{align*}
	After plugging in $e_j(x)$ for for all $x$ into our expression for $\lambda_x$, we again get terms of the form $\sum_j e_j(x)^2$ and cross terms $\sum_{j \ne k} e_j(x)e_k(x)$, similar to the proof of Theorem \ref{thm:id}. Our goal is to combine these terms into a more useful representation. We precisely show this in Lemma \ref{lem:closeness_ev} where we prove that $\E[Z]$ is given by  $\| \vec{\mathbf{e}}_1 + \cdots + \vec{\mathbf{e}}_s\|_2^2$
	where we interpret the vector $\vec{\mathbf{e}}_j \in \mathbb{R}^n$ as the vector with entries $e_j(x) = |q(x)-p_j(x)|$. Note that this is a natural generalization of the quantity $s^2\|q-p\|_2^2$. More formally, we have the following lemma which we prove in Section \ref{sec:closeness_ev}.
	\begin{restatable}{lemma}{lemevCLOSENESS}\label{lem:closeness_ev}
		Let $\{p_i\}_{i=1}^s$ be distributions over $[n]$ that satisfy the structural condition given in Definition \ref{def:structural}. Suppose we draw $\textup{Poi}(1)$ samples from each $p_i$ and $\textup{Poi}(s)$ samples from $q$ and let $T_x$ be the number of times we see element $x \in [n]$ among the samples among the $p_i$, and let $Y_x$ be the number of times we see element $x \in [n]$ among the samples from $q$. Let $Z = \sum_{x \in [n]} (T_x-y_x)^2-T_x-Y_x$. Then
		$$ \E[Z]= \| \vec{\mathbf{e}}_1 + \cdots + \vec{\mathbf{e}}_s\|_2^2$$
		where $\vec{\mathbf{e}}_j \in \mathbb{R}^n$ has coordinates $|q(x)-p_j(x)|$.
	\end{restatable}
	We now give a tight upper bound for the variance of our statistic $Z$. Defining $Z_x = (T_x-Y_x)^2-T_x-Y_x$ and recalling the poissonization method, we see that 
	$$\Var[Z] = \sum_{x \in [n]} \Var[Z_x]  = \sum_{x \in [n]} \E[Z_x^2] - \E[Z_x]^2. $$
	Expanding $Z_x^2$, we get terms involving $\lambda_x^k$ for $k \in \{1,2,3,4\}$. Combining these terms, we again get an upper bound of $\Var[Z]$ in terms of the vectors $\vec{\mathbf{e}}_j$. Formally, we prove the following lemma in Section \ref{sec:closeness_var}.
	
	\begin{restatable}{lemma}{lemvarCLOSENESS} \label{lem:closeness_var}
		Let $\{p_i\}_{i=1}^s$ be distributions over $[n]$ that satisfy the structural condition given in Definition \ref{def:structural}. Suppose we draw $\textup{Poi}(1)$ samples from each $p_i$ and $\textup{Poi}(s)$ samples from $q$ and let $T_x$ be the number of times we see element $x \in [n]$ among the samples among the $p_i$, and let $Y_x$ be the number of times we see element $x \in [n]$ among the samples from $q$. Let $Z = \sum_{x \in [n]} (T_x-y_x)^2-T_x-Y_x$. Then
		$$ \textup{Var}(Z) \le 8s \|q\|_2 \left \| \sum_{j=1}^s \vec{\mathbf{e}}_j \right \|_4^2 + 8 \left \| \sum_{j=1}^s \vec{\mathbf{p}}_j \right \|_2^2 + 4  \left \| \sum_{j=1}^s \vec{\mathbf{e}}_j \right \|_3^3$$
		where $\vec{\mathbf{e}}_j \in \mathbb{R}^n$ has coordinates $|q(x)-p_j(x)|$ and  $\vec{\mathbf{p}}_j \in \mathbb{R}^n$ has coordinates $p_j(x)$.
	\end{restatable}
	We can now proceed to the proof of the theorem in the completeness case. 
	\vspace{2mm}\noindent \textbf{Proof of the Completeness Case:\quad} 
	In this case, Lemma \ref{lem:closeness_ev} gives us $\E[Z] = 0$ and Lemma \ref{lem:closeness_var} gives us $\Var[Z] \le 8s^2\|q\|_2^2$. Therefore by Chebyshev's inequality,
	$$\Pr[|Z| \ge \tau] =  \Pr\left[|Z| \ge \frac{s^2\epsilon^2}{4n} \right] \le \frac{128s^2\|q\|_2^2n^2}{s^4 \epsilon^4} = \frac{128\|q\|_2^2n^2}{s^2\epsilon^4}.$$
	The right hand side of the above inequality can be made arbitrarily small by letting $s = c_1n \|q\|_2/\epsilon^2$ for a sufficiently large constant $c_1$. Due to our randomized flattening procedure of Section \ref{sec:flattening_closeness}, we can assume that $\|q\| = O(1/\sqrt{k})$. Therefore, we can make the above probability bound arbitrarily small by letting $s = c_1n /(\epsilon^2\sqrt{k})$ for a sufficiently large constant $c_1$.
	
	\vspace{2mm}\noindent \textbf{Proof of the Soundness Case:\quad} In this case, 
	$$ \E[Z] = \| \vec{\mathbf{e}}_1 + \cdots + \vec{\mathbf{e}}_s\|_2^2 = \sum_{x \in [n]} \left( \sum_{j=1}^s e_j(x) \right)^2 \ge \frac{1}n \left( \sum_{j=1}^s \sum_{x \in [n]} e_j(x) \right)^2 \ge \frac{s^2\epsilon^2}n $$ where the first inequality is Cauchy-Schwarz and the last inequality follows from our assumption about the error terms $e_j(x)$ in the soundness case. Then by Chebyshev's inequality,
	\begin{align*}
		\Pr\left[ |Z - \E[Z] | \ge \frac{\E[Z]}{4} \right] &\le \frac{16 \Var[Z]}{\E[Z]^2} \\
		& \le  \frac{8s \|q\|_2\left \| \sum_{j=1}^s \vec{\mathbf{e}}_j \right \|_4^2}{\left \| \sum_{j=1}^s \vec{\mathbf{e}}_j \right\|_2^4} +   \frac{8 \left \| \sum_{j=1}^s \vec{\mathbf{p}}_j \right \|_2^2}{\left \| \sum_{j=1}^s \vec{\mathbf{e}}_j \right\|_2^4} +  \frac{4  \left \| \sum_{j=1}^s \vec{\mathbf{e}}_j \right \|_3^3}{\left \| \sum_{j=1}^s \vec{\mathbf{e}}_j \right\|_2^4}.
	\end{align*}
	Note that the right hand side of the above inequality is identical to the right hand side of Inequality \eqref{eq:chebyshevbound} that appears in the probability calculation in the proof of Theorem \ref{thm:id}. Using the identical bounds given there, we arrive at the following inequality:
	$$ \Pr\left[ |Z - \E[Z] | \ge \frac{\E[Z]}{4} \right] \le C\left( \frac{n \|q\|_2}{s \epsilon^2} + \left( \frac{n \|q\|}{s \epsilon^2} \right)^2  + \left( \frac{1}{s \|q\|_2} \right)^2 + \frac{\sqrt n}{s \, \epsilon}\right)$$
	for some constant $C$. Note that if we let $s = c_1n \|q\|_2/\epsilon^2$ for a sufficiently large constant $c_1$ and use the fact that $\|q\|_2 \ge 1/\sqrt{n}$, we have that the above probability is smaller than $1/3$. Again using the randomized flattening procedure of Section \ref{sec:flattening_closeness}, we see that we can let $s = c_1 n /(\epsilon^2 \sqrt{k})$. Hence with probability at least $2/3$, we know $Z \ge 3s^2\epsilon^2/(4n)$ in the soundness case so we reject with probability at least $2/3$, as desired.
\end{proof}
\subsection{Randomized Flattening Procedure}\label{sec:flattening_closeness}
in this section, we present our randomized flattening procedure. Let $k$ be some fixed parameter (which is an input to \textbf{Closeness-Tester}). We show that we can assume that $\|q\|_2^2 = O(1/k)$ for Section \ref{sec:closeness_testing} without loss of generality where $q$ is the distribution that we have i.i.d.\@ sample access to. This procedure is similar to the one used in \cite{DiakonikolasK:2016} for single distribution closeness testing.

First, suppose that we draw $\text{Poi}(k)$ i.i.d.\@ samples from $q$. Then for each $x \in [n]$, define $b_x$ to be the number of instances of element $x \in [n]$ that we see among these samples plus $1$. Note the resemblance between this definition and the one given in the flattening procedure for identity testing in Section \ref{sec:flattening_id}.
Now given a sample $x$ from a distribution $p$ over $[n]$, we can get a sample from the `flattened' distribution $p'$ over
$$ \mathcal{D} = \{(x,y) \mid x \in [n], y \in [b_x] \}$$
by drawing an element from $y \in [b_x]$ uniformly at random and creating the tuple $(x,y)$. This is the flattening procedure that we use for our generalized version closeness testing. Note that the probability mass over $[n]$ placed by $p$ gets `flattened' to be a probability distribution over $\mathcal{D}$, hence the name. The size of this new domain is $n+k$. We can calculate that this procedure preserves the $\ell_1$-distance:
$$ \|q' - p'\|_1 = \sum_{x \in [n]} \sum_{y \in [b_x]} \frac{|q(x)-p(x)|}{|b_x|}  = \sum_{x \in [n]} |q(x)-p(x)| = \|q-p\|_1.$$
Furthermore,
$$ \E[ \|q'\|_2^2] = \E \left[ \sum_{x \in [n]} \sum_{y \in [b_x]} \frac{q(x)^2}{b_x^2} \right] \le \sum_{x \in [n]}  q(x)^2\E \left[ 1/b_x \right].$$
By the poissonization method, we know that $b_x$ is distributed as $1+Z$ where $Z$ is a $\text{Poi}(kq(x))$ random variable. Therefore, similar to~\cite{DiakonikolasK:2016}, we have:
$$ \E\left[ 1/(Z+1) \right] = \E\left[ \int_0^1 s^z \ ds \right] = \int_0^1 \E[s^z] \ ds = \int_0^1 e^{kq(x)(s-1)} \ ds \le \frac{1}{kq(x)}$$
where we have used the probability generating function for a Poisson random variable. This gives us
$$ \E[ \|q'\|_2^2]\le \sum_{x \in [n]} \frac{q(x)}{k} \le \frac{1}k.$$
Hence by Markov's inequality, we can say that $\|q'\|_2^2 = O(1/k)$ holds with an arbitrarily large, constant probability. Now whenever we get a sample over $[n]$, we can use this flattening procedure to draw a sample over $\mathcal{D}$. Furthermore, by using this flattening procedure to draw samples from a slightly larger domain, we can assume that $\|q\|_2^2 = O(1/k)$ in Section \ref{sec:closeness_testing} at the expense of losing a negligible factor in our error probability. 

Note that the size of the larger domain is $O(n+k) = O(n)$ if we pick $k \le n$. Therefore, combining with Theorem \ref{thm:closeness}, we show that we can perform closeness testing with multiple sources by using 
$$O\left(k + \frac{n\|q\|_2}{\epsilon^2}\right) = O\left(k + \frac{n}{\epsilon^2\sqrt{k}}\right) = O\left(\frac{n^{2/3}}{\epsilon^{4/3}} + \frac{\sqrt{n}}{\epsilon^2}\right) $$
samples after optimizing the value of $k$. This sample complexity is optimal since a matching lower bound holds for the single distribution closeness testing setting.
\subsection{Proof of Lemma \ref{lem:closeness_ev}} \label{sec:closeness_ev}
\lemevCLOSENESS*
\begin{proof}
	Note that $T_x$ is a Poisson random variable with parameter $\lambda_x = \sum_{j=1}^s p_j(x)$ and $Y_x$ is a Poisson random variable with parameter $sq(x)$. Let
	$$Z_x = (T_x-Y_x)^2- T_x-Y_x.$$
	We can compute that 
	\begin{align*}
		\E[Z_x] &= \E[T_x^2]-\E[T_x] + \E[Y_x^2]-\E[Y_x]-2\E[Y_x]\E[T_x]  \\
		&= \lambda_x^2 -2sq(x)\lambda_x + s^2q(x)^2.
	\end{align*}
	This is the same as the expected value of the variable $Z_x$ in Lemma \ref{lem:id_ev}. Therefore, the same computations hold and we arrive at
	$$ \E[Z] = \sum_{x \in [n]} \E[Z_x]  = \| \vec{\mathbf{e}}_1 + \cdots + \vec{\mathbf{e}}_s\|_2^2,$$
	as desired.
\end{proof}

\subsection{Proof of Lemma \ref{lem:closeness_var}} \label{sec:closeness_var}
\lemvarCLOSENESS*
\begin{proof}
	Define
	$$Z_x = (T_x-Y_x)^2- T_x-Y_x.$$
	Due to the independence of $T_x$ and $Y_x$, we have
	$$\text{Var}(Z) = \sum_{x \in [n]} \text{Var}(Z_x)  = \sum_{x \in [n]} \E[Z_x^2] - \E[Z_x]^2. $$
	Noting that $T_x$ is Poisson with parameter $\lambda_x = \sum_{j=1}^sp_j(x)$ and $Y_x$ is Poisson with parameter $sq(x)$, we can compute that
	\begin{align*}
		\E[Z_x^2] &= \lambda_x^4 + 4\lambda_x^3(1-sq(x)) + 2 \lambda_x^2(3s^2q(x)^2-2sq(x)+1 ) \\
		& + 4sq(x) \lambda_x(1-sq(x) -s^2q(x)^2) + s^4q(x)^4 + 4s^3q(x)^3+2s^2q(x)^2.
	\end{align*}
	Using the formula for $\lambda_x^k$ for  $k \in \{1,2,3,4\}$ given in Appendix \ref{sec:moment_calcs} and simplifying, we arrive at the following expression:
	\begin{align*}
		&\sum_{x \in [n]} \E[Z_x^2] = 8s \sum_{j = 1}^s \sum_{x \in [n]} q(x)e_j(x)^2 + 8s \sum_{j\ne k} \sum_{x \in [n]} q(x)e_j(x)e_k(x) +8s^2\|q\|_2^2  \\
		&+ 8s \sum_{j=1}^s \sum_{x \in [n]}(-1)^{x \in B}q(x)e_j(x) + 2 \sum_{j=1}^s \sum_{x \in [n]} e_j(x)^2 + 2 \sum_{j \ne k} \sum_{x \in [n]} e_j(x)e_k(x) \\
		&+ 4 \sum_{j=1}^s \sum_{x \in [n]} (-1)^{x \in B}e_j(x)^3 + 12 \sum_{j \ne k}  \sum_{x \in [n]} (-1)^{x \in B}e_j(x)^2e_k(x)  \\
		&+ 4\sum_{j \ne k \ne \ell} \sum_{x \in [n]} (-1)^{x \in B} e_j(x)e_k(x)e_{\ell}(x) + \Bigg ( \sum_{j=1}^s \sum_{x \in [n]} e_j(x)^4 + 6\sum_{j \ne k \ne \ell}e_j(x)^2e_k(x)e_\ell(x)  \\
		&+ 4\sum_{j \ne k}e_j(x)^3e_k(x) + 3\sum_{j \ne k}e_j(x)^2e_k(x)^2 + \sum_{j \ne k \ne \ell \ne t} e_j(x)e_k(x)e_\ell(x)e_t(x) \Bigg ).
	\end{align*}
	Similar to Lemma \ref{lem:id_var}, we have
	$$
	8s \sum_{j = 1}^s \sum_{x \in [n]} q(x)e_j(x)^2 + 8s \sum_{j\ne k} \sum_{x \in [n]} q(x)e_j(x)e_k(x) \le 8s \|q\|_2 \| \vec{\mathbf{e}}_1 + \cdots + \vec{\mathbf{e}}_s \|_4^2,
	$$
	and
	\begin{align*}
		&8s^2\|q\|_2^2 + 8s \sum_{j=1}^s \sum_{x \in [n]}(-1)^{x \in B}q(x)e_j(x) + 2 \sum_{j=1}^s \sum_{x \in [n]} e_j(x)^2 + 2 \sum_{j \ne k} \sum_{x \in [n]} e_j(x)e_k(x) \\
		&\le 8 \| \vec{\mathbf{p}}_1 + \cdots \vec{\mathbf{p}}_s\|_2^2,
	\end{align*}
	and finally,
	\begin{align*}
		&4 \sum_{j=1}^s \sum_{x \in [n]} (-1)^{x \in B}e_j(x)^3 + 12 \sum_{j \ne k}  \sum_{x \in [n]} (-1)^{x \in B}e_j(x)^2e_k(x) \\
		&+ 4\sum_{j \ne k \ne \ell} \sum_{x \in [n]} (-1)^{x \in B} e_j(x)e_k(x)e_{\ell}(x) = 4 \|\vec{\mathbf{e}}_1 + \cdots + \vec{\mathbf{e}}_s \|_3^3.
	\end{align*}
	As in Lemma \ref{lem:id_var}, the last expression inside the parenthesis of $\E[Z_x^2]$ is precisely $$\sum_{x \in [n]} \left( \sum_{j=1}^s e_j(x) \right)^4.$$ Therefore,
	$$ \sum_{x \in [n]} \E[Z_x^2] \le 8s \|q\|_2 \left \| \sum_{j=1}^s \vec{\mathbf{e}}_j \right \|_4^2 + 8 \left \| \sum_{j=1}^s \vec{\mathbf{p}}_j \right \|_2^2 + 4  \left \| \sum_{j=1}^s \vec{\mathbf{e}}_j \right \|_3^3 + \sum_{x \in [n]} \left( \sum_{j=1}^s e_j(x) \right)^4.$$ 
	The same calculation as in Lemma \ref{lem:id_var} gives us 
	$$ \sum_{x \in [n]} \E[Z_x]^2 = \sum_{x \in [n]} \left( \sum_{j=1}^s e_j(x)^2 + \sum_{j \ne k} e_j(x)e_k(x) \right)^2 = \sum_{x \in [n]} \left( \sum_{j=1}^s e_j(x) \right)^4 $$
	so altogether,
	\begin{equation}
	\text{Var}(Z) \le 8s \|q\|_2 \left \| \sum_{j=1}^s \vec{\mathbf{e}}_j \right \|_4^2 + 8 \left \| \sum_{j=1}^s \vec{\mathbf{p}}_j \right \|_2^2 + 4  \left \| \sum_{j=1}^s \vec{\mathbf{e}}_j \right \|_3^3.
	\end{equation}
\end{proof}

\section{Failure of de Finetti's Theorem with Sublinear Number of Samples} \label{sec:definetti}

in this section, we prove Theorem \ref{thm:definettiex}. Recall the statement of Theorem \ref{thm:definettiex}.

\definettiex*

Note that an algorithm can turn samples $Y_1, \cdots, Y_s$ into an exchangeable sequence $X_1, \cdots, X_s$ by permuting randomly. In this section, we given an example of samples $Y_1, \cdots, Y_s$ such that after permuting them to turn them into an exchangeable sequence $X_1, \cdots, X_s$ , the exchangeable sequence is `far' from a mixture of product distributions. The main idea behind Theorem \ref{thm:definettiex} is based on Proposition $31$ in \cite{main_exchangeability}. Essentially, Diaconis and Freedman show in \cite{main_exchangeability} that a Polya's urn process generates an exchangeable sequence that is far from the mixture of any product distributions. This example does not quite work in our case since we would like the  \textit{structural condition} to hold. Therefore, we adapt the Polya's urn idea by partitioning our domain into a `large' set and a `small' set. We then apply a Polya's urn type process on the small set so that no collisions can happen on it. This results in an event $\mathcal{E}_1$ which we use to lower bound the distance from the distribution of our exchangeable sequence to any mixture of product distributions. However, we need an additional event $\mathcal{E}_2$ to deal with the large set. Combining these two events allows us to prove Theorem \ref{thm:definettiex}. This overview is formalized below.

\begin{proof}[Proof of Theorem \ref{thm:definettiex}]
	Let $\epsilon = 1/3$ and let $s$ be the number of samples required by our uniformity tester in Algorithm \ref{alg:uniformity-tester} for this value of $\epsilon$. In particular, $s^2 = Cn$ for some constant $C > 10$ and define $\delta$ as $ \delta = 1/C$. We now construct distributions $\{q_i\}_{i=1}^s$ as follows:
	$$ 
	q_i(x) = \begin{cases}
	1-\delta/20-s/n, &x = 1 \\
	1/n, & x \in  \{2 , \cdots, s+1\} \setminus \{i+1\} \\
	\delta/20, & x  = i+1.
	\end{cases}
	$$
	Note that all distributions are supported only on $\{1, \cdots, s+1\}$. We can check that for large enough $n$, we have $\|q_i-\UU_n\|_1 \ge 1/3$ for all $i$. We then draw sample $Y_i$ independently from $q_i$. Note that $A = \{1, \cdots, s+1\}$ and $B = [n] \setminus \{1, \cdots, s+1\}$ for the \textit{structural condition} in Definition \ref{def:structural}. In other words, all the distributions are larger than uniform on $A$ and smaller than uniform on $B$ so $\{q_i\}_{i=1}^s$ satisfy the conditions of the theorem.
	
	Now let $\{X_i\}_{i=1}^s$ be an exchangeable sequence derived from $\{Y_i\}_{i=1}^s$ (for example, by permuting them randomly). Let $p$ be the distribution of $(X_1, \cdots, X_s)$. Consider the following two events:
	\begin{align*}
		\mathcal{E}_1 &= \text{Event that }X_i =  X_j \in \{2, \cdots, s+1\} \text{ for some } i \ne j \\
		\mathcal{E}_2 &= \text{Event that at least }s(1-\delta^2/2) \text{ of the } X_i\text{'s are equal to }1. 
	\end{align*}
	We first compute $p(\mathcal{E}_1)$. Note that for any $i \ne j$, we have
	$$\mathbb{P}(X_i = X_j \in \{2, \cdots, s+1\}) \le \frac{\delta}{10n} + \frac{s}{n^2}. $$
	Therefore by a union bound, the probability that there exists some $i \ne j$ such that $\mathcal{E}_1$ holds is at most
	$$ \frac{C\delta}{10} + \frac{s^3}{n^2} \le \frac{1}{9} $$
	for sufficiently large $n$. We now compute $p(\mathcal{E}_2)$. We know that  $\mathbb{P}(X_i = 1) \le 1-\delta/20$ for all $i$. Therefore, the number of $1$'s among the $X_i$'s is sum of $s$ Bernoulli random variables with parameters at most $1-\delta/20$ and hence, the expected number of $1$'s is at most $s(1-\delta/20)$. Then by a Chernoff bound, we have $\mathbb{P}(\mathcal{E}_2) \le 1/100$ if we take $n$ (and therefore $s$) sufficiently large. 
	
	Now for a distribution $p_k$ over $[n]$, we let $p_k^s$ be the distribution of $s$ independent picks from $p_k$. 
	Let $M = \sum_k w_k p_k^s$ be any mixture of product distributions $p_k^s$. We wish to show that $\|p-M\|_1 = \Omega(1)$ for any $M$. 
	Note that all of the $\{q_i\}_{i=1}^s$ are only supported on $\{1, \cdots, s+1\}$. Therefore, we can assume without loss of generality that the $p_k$ is also supported only on $\{1, \cdots, s+1\}$ for all $k$. Now consider a single $p_k^s$. We consider two cases. 
	\vspace{2mm}\noindent \textbf{Case 1: $p_k(\{2,\cdots, s+1\}) \ge \delta^2$}. In this case, let $Z$ be the number of collisions among elements in $\{2, \cdots, s+1\}$ if we draw $s$ independent samples from $p_k$. We have
	$$ \E[Z] = \binom{s}2 \sum_{i \in \{2, \cdots, s\}}p_k(i)^2.$$
	Define  $\sum_{i \in \{2, \cdots, s\}}p_k(i)^2 = \|\widetilde{p_k}\|_2^2$. Using standard calculations as in the single distribution uniformity testing case, see \cite{Goldreich2011, Batu:2000}, we can compute
	$$ \Var[Z] \le 4\left( \binom{s}2 \| \widetilde{p_k} \|_2^2 \right)^{3/2}.$$ 
	Hence by Chebyshev's inequality,
	$$
	\mathbb{P}(Z = 0) \le \mathbb{P}(|Z - \E[Z]| \ge \E[Z]) \le \frac{4 \binom{s}2^{3/2} \| \widetilde{p_k} \|_2^3}{\binom{s}2^2 \| \widetilde{p_k} \|_2^4} \le \frac{C'}{s \|\widetilde{p_k} \|_2}
	$$
	for some absolute constant $C'$. Now by Cauchy Schwarz,
	$$ \|\widetilde{p_k}\|_2 \ge \frac{\delta^2}{\sqrt{s}} $$
	so we have $\mathbb{P}(Z = 0) \le 1/100$ for sufficiently large $n$. Therefore, $p_k^s(\mathcal{E}_1) \ge 99/100$. 
	\\
	\vspace{2mm}\noindent \textbf \noindent \textbf{Case 2: $p_k(\{2,\cdots, s+1\}) < \delta^2$.\quad} In this case, we have $p_k(1) \ge 1-\delta^2$. Therefore, if we draw $s$ samples from $p_k$, the number of $1$'s that we will see is at least $s(1-\delta^2)$ in expectation. By Chernoff, the probability that we see less than $s(1-\delta^2/2)$ number of $1$'s is at at most $1/100$ for sufficiently large $n$. Hence, $p_k^s(\mathcal{E}_2) \ge 99/100$.
	\\
	Now note that 
	$$ \| p-M\|_1 = 2 \| p-M\|_{TV}  \ge |P(\mathcal{E}_1) - M(\mathcal{E}_1)| + |p(\mathcal{E}_2) - M(\mathcal{E}_2)|.$$
	We have
	$$|p(\mathcal{E}_1) - M(\mathcal{E}_1)|  \ge \sum_{k \mid p_k \in \text{Case 1}} w_kp_k^s(\mathcal{E}_1) - \frac{1}9 \ge \frac{99}{100} \sum_{k \mid p_k \in \text{Case 1}} w_k - \frac{1}9. $$
	Similarly,
	$$|p(\mathcal{E}_2) - M(\mathcal{E}_2)| \ge   \sum_{k \mid p_k \in \text{Case 2}}w_kp_k^s(\mathcal{E}_2)  - \frac{1}{100} \ge  \frac{99}{100} \sum_{k \mid p_k \in \text{Case 2}} w_k - \frac{1}{100}.$$
	Hence we have that
	$$\| p-M\|_1  \ge \frac{99}{100} \left(\sum_k w_k\right) - \left( \frac{1}9 +  \frac{1}{100} \right) \ge \Omega(1).$$
	Since $M$ was arbitrary, we are done. Hence, $p$ is $\Omega(1)$-far from any mixture of product distributions.
\end{proof}

\paragraph{acknowledgements:} The authors would like to thank Sushruth Reddy, Rikhav Shah, and Greg Valiant for helpful discussions about de Finetti's theorem. The authors would also like to thank Ronitt Rubinfeld for valuable feedback. 

\bibliographystyle{alpha}
\bibliography{paper}

\newcommand{\etalchar}[1]{$^{#1}$}
\begin{thebibliography}{VKVK19}

\bibitem[ABR16]{ABR16}
Maryam Aliakbarpour, Eric Blais, and Ronitt Rubinfeld.
\newblock Learning and testing junta distributions.
\newblock In {\em COLT}, pages 19--46, 2016.

\bibitem[ADK15]{ADK15}
Jayadev Acharya, Constantinos Daskalakis, and Gautam Kamath.
\newblock Optimal testing for properties of distributions.
\newblock In {\em NIPS}, pages 3591--3599, 2015.

\bibitem[ADKR19]{ADKR19}
Maryam Aliakbarpour, Ilias Diakonikolas, Daniel Kane, and Ronitt Rubinfeld.
\newblock Private testing of distributions via sample permutations.
\newblock {\em To appear in NeurIPS}, 2019.

\bibitem[BC17]{Batu2017GeneralizedUT}
Tugkan Batu and Cl{\'e}ment~L. Canonne.
\newblock Generalized uniformity testing.
\newblock {\em 2017 IEEE 58th Annual Symposium on Foundations of Computer
  Science (FOCS)}, pages 880--889, 2017.

\bibitem[BCG19]{BlaisC19}
Eric Blais, Cl{\'e}ment~L. Canonne, and Tom Gur.
\newblock Distribution testing lower bounds via reductions from communication
  complexity.
\newblock {\em ACM Trans. Comput. Theory}, 11(2):6:1--6:37, February 2019.

\bibitem[BFF{\etalchar{+}}01]{BatuFFKRW}
Tugkan Batu, Eldar Fischer, Lance Fortnow, Ravi Kumar, Ronitt Rubinfeld, and
  Patrick White.
\newblock Testing random variables for independence and identity.
\newblock In {\em FOCS}, pages 442--451, 2001.

\bibitem[BFR{\etalchar{+}}00]{Batu:2000}
Tugkan Batu, Lance Fortnow, Ronitt Rubinfeld, Warren~D. Smith, and Patrick
  White.
\newblock Testing that distributions are close.
\newblock In {\em Proceedings of the 41st Annual Symposium on Foundations of
  Computer Science}, FOCS '00, pages 259--, Washington, DC, USA, 2000. IEEE
  Computer Society.

\bibitem[BFR{\etalchar{+}}13]{BFRSW}
Tugkan Batu, Lance Fortnow, Ronitt Rubinfeld, Warren~D. Smith, and Patrick
  White.
\newblock Testing closeness of discrete distributions.
\newblock {\em JACM}, 60(1):4:1--4:25, 2013.

\bibitem[Can15]{canonne2015survey}
Cl{\'e}ment~L Canonne.
\newblock A survey on distribution testing: {Y}our data is big. but is it blue?
\newblock {\em ECCC}, 22:63, 2015.

\bibitem[CDVV14]{ChanDVV14}
Siu{-}on Chan, Ilias Diakonikolas, Paul Valiant, and Gregory Valiant.
\newblock Optimal algorithms for testing closeness of discrete distributions.
\newblock In {\em SODA}, pages 1193--1203, 2014.

\bibitem[Cve12]{maclaurin}
Zdravko Cvetkovski.
\newblock {\em Inequalities theorems, techniques and selected problems}.
\newblock Springer, 2012.

\bibitem[DF80]{main_exchangeability}
Persi Diaconis and David Freedman.
\newblock Finite exchangeable sequences.
\newblock {\em The Annals of Probability}, 8(4):745--764, 1980.

\bibitem[DGPP16]{DiakonikolasGPP16}
Ilias Diakonikolas, Themis Gouleakis, J.~Peebles, and Eric Price.
\newblock Collision-based testers are optimal for uniformity and closeness.
\newblock {\em ECCC}, 23:178, 2016.

\bibitem[DGPP18]{DiakonikolasGPP18}
Ilias Diakonikolas, Themis Gouleakis, John Peebles, and Eric Price.
\newblock Sample-optimal identity testing with high probability.
\newblock In {\em 45th International Colloquium on Automata, Languages, and
  Programming, {ICALP} 2018, July 9-13, 2018, Prague, Czech Republic}, pages
  41:1--41:14, 2018.

\bibitem[Dia77]{definetti_alt}
Persi Diaconis.
\newblock Finite forms of de {F}inetti's theorem on exchangeability.
\newblock {\em Synthese}, 36(2):271--281, Oct 1977.

\bibitem[DK16]{DiakonikolasK:2016}
Ilias Diakonikolas and Daniel~M. Kane.
\newblock A new approach for testing properties of discrete distributions.
\newblock In {\em FOCS}, pages 685--694, 2016.

\bibitem[DKN15]{DiakonikolasKN14}
Ilias Diakonikolas, Daniel~M. Kane, and Vladimir Nikishkin.
\newblock Testing identity of structured distributions.
\newblock In {\em SODA}, pages 1841--1854, 2015.

\bibitem[GGR98]{GGR98}
Oded Goldreich, Shafi Goldwasser, and Dana Ron.
\newblock Property testing and its connection to learning and approximation.
\newblock {\em JACM}, 45:653--750, 1998.

\bibitem[Gol16a]{Goldreich2016TheUD}
Oded Goldreich.
\newblock The uniform distribution is complete with respect to testing identity
  to a fixed distribution.
\newblock {\em Electronic Colloquium on Computational Complexity (ECCC)},
  23:15, 2016.

\bibitem[Gol16b]{goldreich_reduction}
Oded Goldreich.
\newblock The uniform distribution is complete with respect to testing identity
  to a fixed distribution.
\newblock {\em ECCC}, 23, 2016.

\bibitem[Gol17]{GoldreichBook17}
Oded Goldreich.
\newblock {\em Introduction to Property Testing}.
\newblock Cambridge University Press, 2017.

\bibitem[GR11a]{GR00}
Oded Goldreich and Dana Ron.
\newblock On testing expansion in bounded-degree graphs.
\newblock In {\em Studies in Complexity and Cryptography. Miscellanea on the
  Interplay between Randomness and Computation}, pages 68--75. Springer, 2011.

\bibitem[GR11b]{Goldreich2011}
Oded Goldreich and Dana Ron.
\newblock {\em On Testing Expansion in Bounded-Degree Graphs}, pages 68--75.
\newblock Springer Berlin Heidelberg, Berlin, Heidelberg, 2011.

\bibitem[Kal05]{definettibook}
Olav Kallenberg.
\newblock {\em Probabilistic Symmetries and Invariance Principles}.
\newblock Springer New York, 2005.

\bibitem[LR05]{lehmann2005testing}
Erich~L. Lehmann and Joseph~P. Romano.
\newblock {\em Testing statistical hypotheses}.
\newblock Springer Texts in Statistics. Springer, 2005.

\bibitem[LRR13]{LRR13}
Reut Levi, Dana Ron, and Ronitt Rubinfeld.
\newblock Testing properties of collections of distributions.
\newblock {\em Theory of Computing}, 9(8):295--347, 2013.

\bibitem[NP33]{NeymanP}
Jerzy Neyman and Egon~S. Pearson.
\newblock On the problem of the most efficient tests of statistical hypotheses.
\newblock {\em Philosophical Transactions of the Royal Society of London.
  Series A, Containing Papers of a Mathematical or Physical Character},
  231(694-706):289--337, 1933.

\bibitem[Pan08]{Paninski:08}
Liam Paninski.
\newblock A coincidence-based test for uniformity given very sparsely-sampled
  discrete data.
\newblock {\em IEEE TOIT}, 54:4750--4755, 2008.

\bibitem[RRSS07]{RaskhodnikovaRSS:2007}
Sofya {Raskhodnikova}, Dana {Ron}, Amir {Shpilka}, and Adam {Smith}.
\newblock Strong lower bounds for approximating distribution support size and
  the distinct elements problem.
\newblock In {\em 48th Annual IEEE Symposium on Foundations of Computer Science
  (FOCS'07)}, pages 559--569, Oct 2007.

\bibitem[Rub12]{Rub12}
Ronitt Rubinfeld.
\newblock Taming big probability distributions.
\newblock {\em XRDS}, 19(1):24--28, 2012.

\bibitem[TKV17]{TianKV17}
Kevin Tian, Weihao Kong, and Gregory Valiant.
\newblock Learning populations of parameters.
\newblock In I.~Guyon, U.~V. Luxburg, S.~Bengio, H.~Wallach, R.~Fergus,
  S.~Vishwanathan, and R.~Garnett, editors, {\em Advances in Neural Information
  Processing Systems 30}, pages 5778--5787. Curran Associates, Inc., 2017.

\bibitem[Val08a]{Valiant:2008}
Paul Valiant.
\newblock Testing symmetric properties of distributions.
\newblock In {\em Proceedings of the Fortieth Annual ACM Symposium on Theory of
  Computing}, STOC '08, pages 383--392. ACM, 2008.

\bibitem[Val08b]{Valiant08}
Paul Valiant.
\newblock Testing symmetric properties of distributions.
\newblock In {\em STOC}, pages 383--392, 2008.

\bibitem[VKVK19]{vinayak19a}
Ramya~Korlakai Vinayak, Weihao Kong, Gregory Valiant, and Sham Kakade.
\newblock Maximum likelihood estimation for learning populations of parameters.
\newblock In {\em Proceedings of the 36th International Conference on Machine
  Learning (ICML)}, pages 6448--6457. PMLR, 2019.

\bibitem[VV17a]{ValiantV14}
Gregory Valiant and Paul Valiant.
\newblock An automatic inequality prover and instance optimal identity testing.
\newblock {\em SICOMP}, 46(1):429--455, 2017.

\bibitem[VV17b]{VV11}
Gregory Valiant and Paul Valiant.
\newblock Estimating the unseen: {I}mproved estimators for entropy and other
  properties.
\newblock {\em JACM}, 64(6):37:1--37:41, 2017.

\bibitem[WY16]{wu2015chebyshev}
Yihong Wu and Pengkun Yang.
\newblock Chebyshev polynomials, moment matching, and optimal estimation of the
  unseen.
\newblock {\em arXiv preprint arXiv:1504.01227v2}, 2016.

\end{thebibliography}

\appendix

\section*{Appendices}

\section{Moment Calculations of Sections \ref{sec:id_testing} and \ref{sec:closeness_testing} } \label{sec:moment_calcs}
In this section, we calculate the moments of random variables that appear in Section~\ref{sec:id_testing} and Section~\ref{sec:closeness_testing}. 
Suppose we have distributions $q$ and $p_1, \cdots, p_s$ such that there exist subsets $A$ and $B = [n] \setminus A$ with the property that $p_j(x) \ge q(x)$ for all $x \in A$ and all $j$ and $p_j(x) \le q(x)$ for all $x \in B$ and all $j$. Now, let $T_x$ be a Poisson random variable with parameter 
$$ \lambda_x = \sum_{j=1}^s p_j(x) .$$
We compute $\sum_{x \in [n]} \lambda_x^k$ for $k \in \{1,2,3,4\}$. We note that
\begin{equation}\label{eq:lam^1}
\sum_{x \in [n]}\lambda_x= \sum_{j=1}^s \sum_{x \in [n]} (q(x)+(-1)^{x\in B}e_j(x))
\end{equation}
We now compute $\sum_{x \in [n]} \lambda_x^2$. We use the notation $(-1)^{x \in B}$ as follows:
$$(-1)^{x\in B} = 
\begin{cases}
-1 \ &\text{if} \ x \in B \\
\ \ 1 \ &\text{if} \ x \not \in B
\end{cases}.$$
Note that:
$$ \sum_{x \in [n]} \lambda_x^2 = \sum_{x \in [n]} \left( \sum_{j=1}^s p_j(x) \right)^2 = \sum_{x \in [n]} \left( \sum_{j=1}^s p_j(x)^2 + \sum_{j \ne k} p_j(x)p_k(x) \right).$$
For fixed $j$ and $k$, we can compute that: 
$$ p_j(x)^2 = q(x)^2 + 2(-1)^{x \in B}q(x)e_j(x) +e_j(x)^2\,,$$
and we have:
$$ p_j(x)p_k(x) = q(x)^2 + (-1)^{x \in B}(q(x)e_k(x) + q(x)e_j(x)) + e_j(x)e_k(x)\,.$$
Putting everything together, we obtain:
\begin{align*}
\sum_{x \in [n]} \lambda_x^2 = \  &s^2||q||_2^2 + 2s \sum_{j=1}^s \sum_{x \in [n]}(-1)^{x \in B}q(x)e_j(x) \\ 
&+ \sum_{j=1}^s \sum_{x \in [n]} e_j(x)^2 + \sum_{j \ne k} \sum_{x \in [n]}e_j(x)e_k(x) \numberthis \label{eq:lam^2} .
\end{align*}
We now compute $\sum_{x \in [n] } \lambda_x^3$. For a fixed $x$, we have:
$$ \lambda_x^3 = \left( \sum_{j=1}^s p_j(x) \right)^3 = \sum_{j=1}^s p_j(x)^3 + 3\sum_{j \ne k} p_j(x)^2p_k(x) + \sum_{j \ne k \ne \ell} p_j(x)p_k(x)p_{\ell}(x).$$
Now for a fixed $j$,
\begin{align*}
   &p_j(x)^3 = (q(x)+(-1)^{x \in B} e_j(x))^3  \\
   &= q(x)^3 + 3(-1)^{x \in B} q(x)^2 e_j(x) + 3q(x)e_j(x)^2 + (-1)^{x \in B}e_j(x)^3,
\end{align*}
while for fixed $j \ne k$,
\begin{align*}
p_j(x)^2p_k(x) &= (q(x) + (-1)^{x \in B} e_j(x))^2(q(x) + (-1)^{x \in B} e_k(x)) \\
&= q(x)^3 + (-1)^{x \in B}q(x)^2 (2e_j(x) + e_k(x)) + (-1)^{x \in B}e_j(x)^2e_k(x) \\
&+ q(x)(e_j(x)^2 + 2e_j(x)e_k(x)+e_k^2(x)),
\end{align*}
and finally for $j \ne k \ne \ell$, we have: 
\begin{align*}
&p_j(x)p_k(x)p_{\ell}(x) = (q(x) + (-1)^{x \in B} e_j(x))(q(x) + (-1)^{x \in B} e_k(x))(q(x) + (-1)^{x \in B} e_{\ell}(x)) \\
&= q(x)^3 + (-1)^{x \in B}q(x)^2(e_j(x) + e_k(x)+e_{\ell}(x)) + (-1)^{x \in B}e_j(x)e_k(x)e_{\ell}(x) \\
&+ q(x)(e_j(x)e_k(x)+e_k(x)e_{\ell}(x)+e_{j}(x)e_{\ell}(x)).
\end{align*}
Putting everything together, we have
\begin{align*}
\sum_{x \in [n]} \lambda_x^3 &= s^3\|q\|_3^3 + 3s^2 \sum_{j=1}^s \sum_{x \in [n]} (-1)^{x \in B}q(x)^2 e_j(x) + 3 \sum_{j=1}^s \sum_{x \in [n]} q(x)e_j(x)^2 \\
&+ 3s \sum_{j \ne k} \sum_{x \in [n]} q(x)e_j(x)e_k(x) + 3 \sum_{j \ne k} \sum_{x \in [n]} (-1)^{x \in B}e_j(x)^2e_k(x)  \\
&+  \sum_{j =1}^s \sum_{x \in [n]} (-1)^{x \in B} e_j(x)^3  + \sum_{j \ne k \ne \ell} \sum_{x \in [n]} (-1)^{x\in B} e_j(x)e_k(x)e_{\ell}(x) \numberthis \label{eq:lam^3}.
\end{align*}
We now compute $\sum_{x \in [n]} \lambda_x^4$. We have
\begin{align*}
\lambda_x = \left( \sum_{j=1}^s p_j(x) \right)^4 &= \sum_{j=1}^s p_j(x)^4  + 3\sum_{j \ne k} p_j(x)^2p_k(x)^2 + 4 \sum_{j \ne k}p_j(x)^3p_k(x) \\
& + 6 \sum_{j \ne k \ne \ell} p_j(x)^2p_k(x)p_\ell(x)+  \sum_{j \ne k \ne \ell \ne t} p_j(x)p_k(x)p_\ell(x)p_t(x).
\end{align*}
We first analyze $ p_j(x)^4$ for a fixed $j$. We have:
\begin{align*}
p_j(x)^4 = (q(x)+ (-1)^{x \in B}e_j(x))^4  &= q(x)^4 + 4(-1)^{x\in B}q(x)^3e_j(x) + 6q(x)^2e_j(x)^2 \\
&+ 4(-1)^{x \in B}q(x)e_j(x)^3 + e_j(x)^4.
\end{align*}
Then for fixed $j \ne k$, we have:
\begin{align*}
p_j(x)^2p_k(x)^2 &=  (q(x)+ (-1)^{x \in B}e_j(x))^2 (q(x)+ (-1)^{x \in B}e_k(x))^2 \\
&= q(x)^4 + (-1)^{x \in B}q(x)^3(2e_j(x)+2e_k(x)) \\
&+ q(x^2(e_j(x)^2 + 4e_j(x)e_k(x) + e_k(x)^2) \\
&+ (-1)^{x \in B}q(x)(2e_j(x)^2e_k(x)+2e_j(x)e_k(x)^2) + e_j(x)^2e_k(x)^2 \,,
\end{align*}
and, we can get:
\begin{align*}
p_j(x)^3p_k(x) &=  (q(x)+ (-1)^{x \in B}e_j(x))^3 (q(x)+ (-1)^{x \in B}e_k(x)) \\
&= q(x)^4 + (-1)^{x \in B}q(x)^3(e_j(x)+e_k(x)) + q(x)^2(3e_j(x)^2 + 3e_j(x)e_k(x)) \\
&+ (-1)^{x \in B}q(x)(e_j(x)^3 + 3e_j(x)^2e_k(x)) + e_j(x)^3e_k(x).
\end{align*}
Furthermore, for fixed $j \ne k \ne \ell $, we have:
\begin{align*}
&p_j(x)^2p_k(x)p_\ell(x) \\
&=  (q(x)+ (-1)^{x \in B}e_j(x))^2 (q(x)+ (-1)^{x \in B}e_k(x))(q(x)+ (-1)^{x \in B}e_\ell(x)) \\
&= q(x)^4 +(-1)^{x \in B} q(x)^3(2e_j(x)+e_\ell(x)+e_k(x)) \\
&+q(x)^2(e_j(x)^2+2e_j(x)e_k(x)+2e_j(x)e_\ell(x)+e_k(x)e_\ell(x)) \\
&+(-1)^{x \in B}q(x)(e_j(x)^2e_\ell(x) + e_j(x)^2e_k(x) + 2e_j(x)e_k(x)e_\ell(x)) \\
&+ e_j(x)^2e_k(x)e_\ell(x).
\end{align*}
Finally, for fixed $j \ne k \ne \ell \ne t$, we have: 
\begin{align*}
&p_j(x)p_k(x)p_\ell(x)p_t(x) \\
&=  (q(x)+ (-1)^{x \in B}e_j(x)) (q(x)+ (-1)^{x \in B}e_k(x)) \\
&\cdot (q(x)+ (-1)^{x \in B}e_\ell(x)) (q(x)+ (-1)^{x \in B}e_t(x))  \\
&= q(x)^4 +(-1)^{x \in B}q(x)(e_j(x)+e_k(x)+e_\ell(x)+e_t(x)) \\
&+ q(x)^2 (e_j(x)e_k(x) + e_j(x)e_\ell(x) + e_j(x)e_t(x) + e_k(x)e_\ell(x) + e_k(x)e_t(x) + e_\ell(x)e_t(x) ) \\
&+(-1)^{x \in B}q(x)(e_j(x)e_k(x)e_\ell(x) + e_j(x)e_k(x)e_t(x) + e_j(x)e_t(x)e_\ell(x) + e_k(x)e_\ell(x)e_t(x)) \\
&+e_j(x)e_k(x)e_\ell(x)e_t(x).
\end{align*}
Altogether, we have: 
\begin{align*}
&\sum_{x \in [n]} \lambda_x^4 \\
&= s^4 \|q \|_4^4 + \sum_{x \in [n]}  \sum_{j=1}^s e_j(x)^4+6s^2 \sum_{j=1}^s \sum_{x \in [n]} q(x)^2e_j(x)^2 + 4s \sum_{j=1}^s \sum_{x \in [n]}(-1)^{x \in B}q(x)e_j(x)^3 \\
&+ 4s^3\sum_{j=1}^s \sum_{x \in [n]}(-1)^{x \in B}q(x)^3e_j(x) + 6s^2\sum_{j \ne k} \sum_{x \in [n]} q(x)^2e_j(x)e_k(x)  \\
&+ 12s \sum_{j \ne k} \sum_{x \in [n]} (-1)^{x \in B} q(x)e_j(x)^2e_k(x) + 4s \sum_{j \ne k \ne \ell} \sum_{x \in [n]} (-1)^{x \in B} q(x)e_j(x)e_k(x)e_\ell(x) \\
&+ 6 \sum_{j \ne k \ne \ell} e_j(x)^2e_k(x)e_\ell(x) + 4 \sum_{j \ne k} \sum_{x \in [n]}e_j(x)^3e_k(x) + 3  \sum_{j \ne k} \sum_{x \in [n]}e_j(x)^2e_k(x)^2 \\
&+ \sum_{ j \ne k \ne \ell \ne t} e_j(x)e_k(x)e_\ell(x)e_t(x) \numberthis \label{eq:lam^4}.
\end{align*}

\end{document}